 \documentclass[10pt]{article}
 \usepackage[top=1in,bottom=1in,left=1.5in,right=1.5in]{geometry}
 \usepackage{amsmath}
  \usepackage[dvips]{pstricks} 
  \usepackage{pst-node}
  \usepackage[dvips]{graphicx}
  \usepackage[standard,amsmath,thmmarks]{ntheorem}
  \newcommand{\twocases}[5]{#1 =
               \left\{
                \begin{alignedat}{2}
                 &#2 \quad & &\text{if } #3\\
                 &#4 \quad & &\text{if } #5
                \end{alignedat}
               \right.}
  
  \newcommand{\C}{\mathbb{C}}
  
  \newcommand{\N}{\mathbb{N}}
  \newcommand{\R}{\mathbb{R}}
  \newcommand{\Z}{\mathbb{Z}}
  \renewcommand{\i}{\mathbf{i}}
  \renewcommand{\j}{\mathbf{j}}
  \renewcommand{\k}{\mathbf{k}}
  \renewcommand{\a}{\mathbf{a}}
  \renewcommand{\b}{\mathbf{b}}
  \renewcommand{\d}{\mathbf{d}}
  \newcommand{\e}{\mathbf{e}}
  
  \newcommand{\m}{\mathbf{m}}
  \newcommand{\n}{\mathbf{n}}
  \newcommand{\p}{\mathbf{p}}
  \newcommand{\q}{\mathbf{q}}
  \newcommand{\br}{\mathbf{r}}
  \newcommand{\s}{\mathbf{s}}
  \newcommand{\U}{\mathbf{U}}
  \renewcommand{\u}{\mathbf{u}}
  \renewcommand{\v}{\mathbf{v}}
  \newcommand{\w}{\mathbf{w}}
  \newcommand{\x}{\mathbf{x}}
  \newcommand{\y}{\mathbf{y}}
  \newcommand{\z}{\mathbf{z}}
  \newcommand{\0}{\mathbf{0}}
  \newcommand{\1}{\mathbf{1}}
  \newcommand{\Gam}{\mathbf{\Gamma}}

  \newcommand{\bpi}{\Pi_\Gam}
  \newcommand{\bM}{\mathbf{M}}
  \newcommand{\bc}{\mathbf{c}}
  \newcommand{\bC}{\mathbf{C}}

  \newcommand{\cC}{\mathcal{C}}

  \newcommand{\cF}{\mathcal{F}}

  \newcommand{\cP}{\mathcal{P}}

  \newcommand{\lan}{\langle}
  \newcommand{\ran}{\rangle}
  \newcommand{\an}[1]{\lan#1\ran}
  \newcommand{\hs}{\hspace*{\parindent}}
  \newcommand{\cl}{\mathop{\mathrm{cl}\;}\nolimits}
  \newcommand{\dom}{\mathop{\mathrm{dom}}\nolimits}
  \newcommand{\tr}{\mathop{\mathrm{tr}}\nolimits}

  \newcommand{\trans}{^\top}

  \newcommand{\perio}{\mathrm{per}}
  \newcommand{\pres}{\mathrm{pres}}
  \newcommand{\hpres}{\mathrm{hpres}}
  \newcommand{\conv}{\mathrm{conv\;}}
  \newcommand{\inter}{\mathop{\mathrm{int\;}}\nolimits}
  \newcommand{\Cl}{\mathrm{cl\;}}
  \newcommand{\Der}{\mathop{\mathrm{diff}}\nolimits}
  \newcommand{\topo}{\mathrm{top}}
  \newcommand{\vol}{\mathrm{vol}}

  \newcommand{\ri}{\mathop{\mathrm{ri\;}}\nolimits}

  \newcommand{\set}[1]{\{#1\}}
  \newtheorem{theo}{\bfseries \hs Theorem}[section]
  \newtheorem{defn}[theo]{\bfseries \hs Definition}
  \newtheorem{prop}[theo]{\bfseries \hs Proposition}
  \newtheorem{corol}[theo]{\bfseries \hs Corollary}
  
  \newtheorem{exam}[theo]{\bfseries \hs Example}
  
  \numberwithin{equation}{section} 

\begin{document}
  \title{The pressure, densities and first order phase transitions\\associated with multidimensional SOFT}
  \author{
  Shmuel Friedland\thanks{Part of this paper was done while this
  author was New Directions Visiting Professor, AY 2003/4, Institute of Mathematics
  and its Applications, University of Minnesota, Minnesota, MN 55455-0436.}\\
  \texttt{friedlan@uic.edu}
                   \and
  Uri N. Peled
  \\
  \texttt{uripeled@uic.edu}
  }
  \date{Department of Mathematics, Statistics, and Computer Science,\\
        University of Illinois at Chicago\\
        Chicago, Illinois 60607-7045, USA\\
        May 17, 2010}

 \maketitle

 \begin{abstract}
 We study theoretical and computational properties of the pressure function for
 subshifts of finite type on the integer lattice $\Z^d$, multidimensional SOFT, which are called
 Potts models in mathematical physics.
 We show that the pressure is Lipschitz and convex.  We use the properties of convex functions
 in several variables  to show rigorously that the phase transitions of the first order correspond 
 exactly to the points where the pressure is not differentiable.
 We give computable upper and lower bounds for the pressure, which can be arbitrary close to the values
 of the pressure given a sufficient computational power.
 We apply our numerical methods to confirm Baxter's
 heuristic computations for two dimensional monomer-dimer model, and to compute the pressure
 and the density entropy as functions of two variables for the two dimensional monomer-dimer model.
 \end{abstract}


 \section{Introduction}

 The most celebrated models in statistical mechanics are the Ising
 models, introduced by Ising in \cite{Isi}, and their
 generalizations to Potts models \cite{Pot}.  Usually, the
 one-dimensional Ising or Potts models admit a closed-form
 analytical solution and do not exhibit the \emph{phase transition}
 phenomenon, as in the case of the original work of Ising for
 ferromagnetism. The importance of Ising models was demonstrated by
 Onsager's closed-form solution for the two-dimensional
 ferromagnetism model in the zero-field case \cite{Ons}, which does
 exhibit phase transition at exactly one temperature.
 Unfortunately, there are only a handful of known closed-form
 solutions for two-dimensional Potts models, including the dimer
 problem due to Fisher, Kasteleyn and Temperley \cite{Fis},
 \cite{Kas}, \cite{TF}; residual entropy of square ice by Lieb
 \cite{Lie}; hard hexagons by Baxter \cite{Bax2}. See also
 \cite{Bax3}.

 Thus, most of the interesting Potts models, in particular all
 problems in dimension $3$ and up, are treated by ad hoc asymptotic
 expansions or by some kind of numerical solutions, in many
 circumstances with the help of Monte Carlo simulations, which
 usually have a heuristic basis.  The aim of this paper is to
 introduce a new mathematical foundation to this subject, which also
 gives rise to reliable numerical methods, using converging upper
 and lower bounds, for computing the \emph{pressure} and its
 derivatives for known quantities in statistical mechanics. In
 principle, these quantities can be computed to any accuracy given
 sufficient computing power.  The first-order phase transition is
 manifested by a jump in a corresponding directional derivative of
 the pressure, which can be detected within the given precision of the
 computation.  Our approach to the phase transition  is significantly simpler
 than the approaches using the Gibbs equilibrium measures corresponding to the
 pressure, e.g. \cite{AKW,Com,Isr,Rue}.
 In models with one variable, the situation
 is relatively well understood by physicists.  The basic argument of phase transition in Ising
 model is due to Peierls \cite{Pei}.  For more modern account of the physicist's approach see \cite
 [pp' 59]{Gri}.

 We now introduce the main ideas of this paper as nontechnically as
 possible.  Assume that we have a standard lattice $\Z^d$,
 consisting of points in $d$-dimensional space $\R^d$ with integer
 coordinates, which we call \emph{sites}.  Each site $\i =
 (i_1,\ldots,i_d)\trans$ is occupied by exactly one particle, or
 color, out of the set $\an{n} := \{1,\ldots,n\}$ of $n$ distinct
 colors (if we do not insist that every site be occupied, we agree
 to use color $n$ for an unoccupied site). In general, one has a
 local type of restriction on the allowed configurations of the
 colors, which is called a \emph{subshift of finite type}, or
 \emph{SOFT}, known as the hard-core model in physics terminology.
 The exact definition of a SOFT is given in the next section. For an
 example of SOFT, consider the residual entropy of square ice
 studied in \cite{Lie}. This entropy is the exponential growth rate
 of the number of colorings of increasing sequences of squares in
 $\Z^2$ with $n = 3$ colors, subject to the local restriction that
 no two adjacent sites receive the same color.  More generally, we
 consider a nonempty \emph{near neighbor SOFT} (NNSOFT), specified
 by a $d$-tuple $\Gam = (\Gamma_1,\ldots,\Gamma_d)$, where each
 $\Gamma_k \subseteq \an{n} \times \an{n}$ is a digraph whose set of
 vertices is the set $\an{n}$ of colors. Two adjacent sites $\i$ and
 $\i + \e_k$, where $\e_k = (\delta_{1k},\ldots,\delta_{dk})\trans$,
 are allowed to receive the colors $p$ and $q$ respectively only if
 $(p,q)\in \Gamma_k$. We denote the set of all allowed colorings in
 this NNSOFT by $C_\Gam(\Z^d)$.

 We assume for simplicity of the exposition that the Hamiltonian
 of a particle of color $i$ is $u_i \in \R$.
 If this is not the case, as for the Ising model or the monomer-dimer
 model, there is a way to reduce such a model to our model by enlarging
 the number of colors.  We show how to carry out this reduction for the monomer-dimer
 model.

 For $\m = (m_1,\ldots,m_d) \in
 \N^d$, let $\an{\m}$ denote the $d$-dimensional box $\an{m_1}
 \times \cdots \times \an{m_d}$. Let $\phi : \an{\m} \to \an{n}$ be a
 coloring $\an{\m}$ with $n$ colors, i.e., an ensemble of $\vol(\m)
 := m_1 \cdots m_d$ particles of $n$ kinds occupying the sites in
 $\an{\m}$. Let $c_i(\phi)$ be the number of sites in $\an{\m}$
 colored with color $i$. Let $\bc(\phi) = (c_1(\phi),\ldots,c_n(\phi))\trans$
 and $\u  =(u_1,\ldots,u_n)\trans \in \R^n$.  Then the Hamiltonian of
 the system $\phi$ is equal to $\bc(\phi)\trans \u$. The \emph{grand
 partition function} corresponding to the set $C_\Gam(\an{\m})$ of
 all colorings $\phi : \an{\m} \to \an{n}$ allowed by $\Gam$ is
 given by
 \begin{equation}\label{gpf}
 Z_\Gam(\m,\u)  :=\sum_{\phi \in C_\Gam(\an{\m})} e^{\bc(\phi)\trans \u}.
 \end{equation}
 It is well-known that $\log Z_\Gam(\m,\u)$ is a convex function.
 Furthermore, the multisequence $\log Z_\Gam(\m,\u)$, $\m \in \N^d$ is
 subadditive in each coordinate of $\m$.  Hence the following
 limit exists
 \begin{equation}\label{pressdef}
 P_\Gam(\u) := \lim_{\m \to \infty} \frac{\log
 Z_\Gam(\m,\u)}{\vol(\m)},
 \end{equation}
 where $\m \to \infty$ means $m_j \to \infty$ for all $j \in
 \an{d}$. This limit is called the \emph{pressure} function. The
 value $h_{\Gam} := P_\Gam(\0)$ is the (free) entropy of the
 corresponding SOFT, and our previous paper \cite{FP} was devoted to
 the theory of its computation. The function $P_\Gam(\cdot)  :\R^n
 \to \R$ is a Lipschitz convex function. Hence it is continuous and
 subdifferentiable everywhere, and differentiable almost everywhere.
 Assume that $P_\Gam(\cdot)$ is differentiable at $\u$ with gradient
 vector $\p(\u)=(p_1(\u),\ldots,p_n(\u))\trans$. Then $\p(\u)$ is a
 probability vector, and $p_i(\u)$ is the relative frequency, or
 proportion, of color $i$ corresponding to the pressure
 $P_\Gam(\u)$. We show that the points $\u$ where $P_\Gam(\u)$ is not
 differentiable correspond to phase transitions of the \emph{first order}, i.e., these are
 points $\u$ where the proportions of the colors are not unique.

 Let $\Pi_n \subseteq \R^n$ denote the simplex of probability
 vectors. Assume that $\p \in \Pi_n$ is a limiting color proportion
 vector for some multisequence of configurations in
 $C_\Gam(\an{\m})$, $\m \to \infty$.  Then one can define the
 \emph{density entropy} $h_\Gam(\p)$ as the maximal exponential growth
 rate of the number of configurations, the maximum being taken over
 all multisequences whose color proportion vector tends to
 $\p$. (See for example \cite{Ha1} for the special case of the
 monomer-dimer configurations.)
 We denote by $\Pi_\Gam\subseteq\Pi_n$ the compact set of all
 limiting color proportion vectors.

 Let $P^*_\Gam(\cdot) : \R^n \to
 [-h_\Gam,\infty]$ be the \emph{conjugate} of $P_\Gam(\cdot)$, which
 is called the Legendre-Fenchel transform in the case of
 differentiable convex functions \cite{Arn,Roc}. Recall that $P^*_\Gam(\cdot)$ is a
 convex function. We show that for a limiting color proportion vector
 $\p$, that is also a subgradient of the pressure function
 somewhere, $h_\Gam(\p) = -P^*_\Gam(\p)$.  Thus $h_\Gam$ is a
 concave function on each convex set of such vectors $\p$ in $\Pi_\Gam$.

 We next show that in many SOFT arising in physical models, the set
 $\Pi_\Gam$ is convex and the function $h_\Gam:\Pi_\Gam\to \R_+$ is concave.
 A simple example is as follows.
 Assume that our SOFT given by $\Gam$ has a friendly color, say $n$.
 This is, in each digraph $\Gamma_k$ the vertices $n$ and $i$ connected in both directions,
 i.e. $(n,i),(i,n)\in\Gamma_k$, for $i=1,\ldots,n$ and $k=1,\ldots,d$.
 Then $\Pi_\Gam$ is convex and $h_\Gam|\Pi_\Gam$ is concave.
 The hard core model has a friendly color.  The monomer-dimer model has essentially a
 friendly color, which corresponds to the dimer, hence
 $\Pi_\Gam=\Pi_{d+1}$ and $h_\Gam|\Pi_{d+1}$ is concave.
 These results can be viewed as generalizations of the result of Hammersley
 \cite{Ha1}.

 For numerical computations of the pressure one needs to have lower bounds for the pressure, which converge
 to the pressure in the limit.  (The convergent upper bounds are given by $ \frac{\log
 Z_\Gam(\m,\u)}{\vol(\m)}$, since
 the multisequence $\log Z_\Gam(\m,\u)$, $\m \in \N^d$ is
 subadditive in each coordinate of $\m$.)  We extend the results in \cite{Fr1,FP}
 to give lower convergent bounds if
 at least $d-1$ digraphs out of $\Gamma_1,\ldots,\Gamma_d$ are symmetric.
 (A digraph $\Gamma\subseteq \an{n}\times \an{n}$ is called symmetric, (reversible), if the diedge $(i,j)$ is in $\Gamma$
 whenever $(j,i)$ is in $\Gamma$.)  This condition holds for most of the known physical models.
 In this paper we show how to apply the computational methods developed in \cite{FP} to the pressure.
 We demonstrate the applications of our methods to the two dimensional monomer-dimer model on $\Z^2$.
 First we confirm the heuristic computations of Baxter \cite{Bax1}.  Second, we find numerically a number of values
 of the pressure function $P_2(v_1,v_2)$ and the value of the density entropy $\bar h_2(p_1,p_2)$ for dimers with densities $p_1,p_2$
 in the directions $x_1,x_2$ respectively.  In Figures 1 and 2 we give the plots of these functions.
 These computations go beyond the known computations of \cite{Bax1,HM},
 where one considers the total density of dimers, (which reduce to the computations of functions of one variable).

 We hope to show that the methods of this paper can be applied to other interesting models.
 We already know that our approach works for the numerical computation of the pressure function for
 the 2D and 3D Ising models in external magnetic field.  (It is similar to the monomer-dimer model
 we study here.)  We plan to study if our numerical computations are precise enough to discover the
 second order phase transitions, which occurs in multidimensional Ising models.

 We now survey briefly the contents of the paper. In
 Section~\ref{softnnsoftandpressure} we describe in details SOFT,
 NNSOFT and the pressure function. We also show that in the
 one-dimensional case, the pressure is the logarithm of the spectral
 radius of a corresponding nonnegative matrix.  In
 Section~\ref{sec:SymmetricNNSOFT} we show that under certain
 symmetry (reversibility) assumptions on $d-1$ digraphs among
 $\Gamma_1,\ldots,\Gamma_d$, we have computable converging upper and
 lower bounds for the pressure.
 In Section~\ref{densityentropy} we relate certain properties of a convex function,
 as differentiability and its conjugate $P^*_\Gam(\p)$, (the Legendre-Fenchel transform),
 to the physical quantities associated with a given SOFT, i.e. the corresponding Potts model.
 In particular we show that the points where the
 pressure $P_\Gam(\cdot)$ is differentiable correspond to unique
 color frequency vectors.  On the other hand the points were the pressure is not differentiable correspond to
 the phase transition of first order, since to this value of $\u$ correspond at least two different color frequencies.
 We also relate the density entropy $h_\Gam(\p)$ to
 the conjugate function $P^*_\Gam(\p)$. In
 Section~\ref{sec:onedimensionalsoft} we apply the results of Section~\ref{densityentropy}
 to a one-dimensional SOFT.  The importance of one-dimensional SOFT is due to the fact that our approximations
 of the pressure are obtained by using the exact results on one-dimensional SOFT.
 Section~\ref{sec:weightedmonomerdimertiligs} we apply some of our results in Section~\ref{densityentropy}
 to the monomer-dimer model in $\Z^d$.  We also relate our results to the works of Hammersley and Baxter \cite{Ha1,Bax1}.
 This is done by using the fact that the monomer-dimer model in $\Z^d$ can be realized as SOFT with $2d+1$ colors \cite{Fr2,FP}.
 As we pointed out in \cite{FP} this SOFT does not have symmetric properties, and hence can not be used for computation.
 In Section~\ref{sec:symencod} we use the symmetric encoding of the monomer-dimer model developed in \cite{FP},
 to obtain computer upper and lower bounds for the pressure function.  In Section~\ref{sec:graphs} we apply our techniques
 to the computations of two dimensional pressure and density entropy for the monomer-dimer model in $\Z^2$.

 \section{SOFT, NNSOFT and Pressure}\label{softnnsoftandpressure}

 We use the notation $\an{r} := \{1,\ldots,r\}$ for $r \in \N :=
 \{1,2,3,\ldots\}$, and for $\m = (m_1,\ldots,m_d) \in \N^d$,
 $\an{\m} := \an{m_1} \times \cdots \times \an{m_d}$ denotes a box
 with volume $\vol(\m) := m_1\cdots m_d$. Then
 $\an{n}^{\an{\m}}$ is the set of all colorings $\phi: \an{\m} \to
 \an{n}$ of $\an{\m}$ with colors from $\an{n}$. We denote by
 $c(\phi)_i:= \# \phi^{-1}(i)$ the number of sites in $\an{\m}$
 colored with the color $i \in \an {n}$, and let $\bc(\phi) :=
 (c(\phi)_1,\ldots,c(\phi)_n)\trans$. Similarly, with $\Z :=
 \{0,\pm 1,\pm 2,\ldots\}$, $\an{n}^{\Z^d}$ is the set of all
 colorings $\phi: \Z^d \to \an{n}$ of $\Z^d$ with colors from
 $\an{n}$. Given a $d$-digraph $\Gam =
 (\Gamma_1,\ldots,\Gamma_d)$ on $\an{n} \times \an{n}$, let
 $C_\Gam(\Z^d) \subseteq \an{n}^{\Z^d}$ be the set of all
 \emph{$\Gam$-colorings}, namely colorings $\phi =
 (\phi_{\m})_{\m \in \Z^d} \in \an{n}^{\Z^d}$ such that for each
 $\i \in \Z^d$ and $k \in \an{d}$, $(\phi_\i,\phi_{\i+\e_k}) \in
 \Gamma_k$, where $\e_k$ is the unit vector with $k$th component
 equal to $1$. In ergodic theory, the set $C_\Gam(\Z^d)$ is called a
 \emph{nearest neighbor subshift of finite type (NNSOFT)}.

 A general SOFT can be described as follows. Let $\bM \in \N^d$
 and a nonempty subset $\cP \subseteq \an{n}^{\an{\bM}}$ be given.
 Every element $a \in \cP$ is viewed as an allowed coloring
 (configuration) of the box $\an{\bM}$ with $n$ colors. For $\i \in
 \Z^d$, we define the shifted coloring $\tau_{\i}(a)$ of $a \in
 \cP$ as the coloring of the shifted box $\an{\bM} + \i$ that
 gives to the site $\x + \i$ the same color that $a$ gives to $\x
 \in \an{\bM}$. We denote by $\tau_{\i}(\cP)$ the set
 $\{\tau_{\i}(a) : a \in \cP\}$, and regard it as the set of
 allowed colorings of $\an{\bM} + \i$. A coloring $\phi \in
 \an{n}^{\Z^d}$ is called a \emph{$\cP$-state} if for each $\i \in
 \Z^d$ the restriction of $\phi$ to $\an{\bM} + \i$ is in
 $\tau_{\i}(\cP)$. We denote by $\an{n}^{\Z^d}(\cP)$ the set of
 all $\cP$-states. In ergodic theory, the set $\an{n}^{\Z^d}(\cP)$
 is called a \emph{subshift of finite type (SOFT}) \cite{Sc}.

 Each NNSOFT $C_\Gam(\Z^d)$ is a special kind of SOFT obtained by
 letting $\bM = (2,\ldots,2)$ and $\cP$ the set of all colorings
 $\phi \in \an{n}^{\an{\bM}}$ such that $\i, \i + \e_k \in \an{\bM}$
 imply $(\phi_{\i},\phi_{\i + \e_k}) \in \Gamma_k$. Conversely
 \cite{Fr1}, each SOFT $\an{n}^{\Z^d}(\cP)$ can be encoded as an
 NNSOFT $C_\Gam(\Z^d)$, where $\Gam =
 (\Gamma_1,\ldots,\Gamma_d)$ is defined as follows. Take $N =
 \#\cP$ and use a bijection between $\cP$ and $\an{N}$. The digraph
 $\Gamma_k \subseteq \an{N} \times \an{N}$ is defined so that for
 $a,b \in \cP$ we have $(a,b) \in \Gamma_k$ if and only if there
 is a configuration $\phi \in \an{n}^{\an{\bM + \e_k}}$ such that
 the restriction of $\phi$ to $\an{\bM}$ is $a$ and the
 restriction of $\phi$ to $\an{\bM} + \e_k$ is $\tau_{\e_k}(b)$.
 Because of this equivalence, we will be dealing here with NNSOFT
 only.

 In the sequel we will take $\limsup$ and $\liminf$ of real
 multisequences $(a_{\m})_{\m \in \N^d}$ as $\m \to \infty$. In order
 to be clear, we define these here and observe that they are limits
 of subsequences \cite{FP}. We also define the limit of real multisequence in
 terms of $\limsup$ and $\liminf$, which is equivalent to other
 definitions in the literature.

 \begin{defn}\label{def:multilimsup}
 Let $(a_{\m})_{\m \in \N^d}$ be a multisequence of real numbers.
 Then
 \begin{enumerate}
 \item[(a)]
 $\limsup_{\m \to \infty} a_\m$ is defined as the supremum
 (possibly $\pm \infty$) of all numbers of the form $\limsup_{q
 \to \infty} a_{\m_q}$, where $(\m_q)_{q \in \N}$ is a sequence in
 $\N^d$ satisfying $\lim_{q \to \infty} \m_q = \infty$, i.e.,
 $\lim_{q \to \infty} (\m_q)_k = \infty$ for each $k \in \an{d}$.
 We define $\liminf_{\m \to \infty} a_\m$ similarly.
 \item[(b)]
 $\lim_{\m \to \infty} a_\m = \alpha$ means
 $\limsup_{\m \to \infty} a_\m = \liminf_{\m \to \infty} a_\m =
 \alpha$.
 \end{enumerate}
 \end{defn}
 As in \cite{FP}, given an NNSOFT $C_\Gam(\Z^d)$ and $\m \in
 \N^d$, we denote by $C_\Gam(\an{\m})$ the set of all colorings $\phi \in
 \an{n}^{\an{\m}}$ such that $\i, \i + \e_k \in \an{\m}$ imply
 $(\phi_{\i},\phi_{\i + \e_k}) \in \Gamma_k$. Similarly, we denote
 by $C_{\Gam,\topo}(\an{\m}) \subseteq C_\Gam(\an{\m})$ the projection of
 $C_\Gam(\Z^d)$ on $\an{\m}$, i.e., the set of colorings in
 $\an{n}^{\an{\m}}$ that can be extended to colorings in
 $C_\Gam(\Z^d)$, and by $C_{\Gam,\perio}(\an{\m}) \subseteq C_{\Gam,\topo}(\an{\m})$
 the set of periodic $\Gam$-colorings with period $\m$, i.e.,
 the set of colorings in $\an{n}^{\an{\m}}$ that can be extended
 to colorings in $C_\Gam(\Z^d)$ with period $\m$. For a weight
 vector $\u=(u_1,\ldots,u_n)\trans \in \R^n$ on the colors, we define
 \begin{align}
  Z_\Gam(\m,\u) &:= \sum_{\phi\in C_\Gam(\an{\m})}e^{\bc(\phi)\trans \u},\label{wmufor}\\
  Z_{\Gam,\topo}(\m,\u) &:= \sum_{\phi\in C_{\Gam,\topo}(\an{\m})}e^{\bc(\phi)\trans \u},\nonumber\\
  Z_{\Gam,\perio}(\m,\u) &:= \sum_{\phi\in C_{\Gam,\perio}(\an{\m})}e^{\bc(\phi)\trans \u}\nonumber.
 \end{align}
 As usual, a summation over an empty set is understood as $0$.
 Obviously
 \begin{align}
  \# C_\Gam(\an{\m}) &= Z_\Gam(\m,\0),\nonumber\\
  \# C_{\Gam,\topo}(\an{\m}) &= Z_{\Gam,\topo}(\m,\0),\nonumber\\
  \# C_{\Gam,\perio}(\an{\m}) &= Z_\perio(\m,\0).\nonumber
 \end{align}
 A function $f(\u) \geq 0$ on $\R^n$ is called \emph{log-convex}
 when $\log f(\u)$ is convex. (The zero function is by definition
 log-convex.) Recall that the log-convex functions are closed under
 linear combinations with nonnegative coefficients \cite{Kin}. Since
 for $\bc\in \R^n$ the function $e^{\bc\trans \u}$ is log-convex,
 the weighted sums $Z_\Gam(\m,\u)$, $Z_{\Gam,\topo}(\m,\u)$,
 $Z_{\Gam,\perio}(\m,\u)$ are log-convex functions of $\u$ for each
 $\m\in \N^d$. As in \cite{FP} it follows that for a fixed $\u$,
 $\log Z_\Gam(\m,\u)$ and $\log Z_{\Gam,\topo}(\m,\u)$ are subadditive in
 each coordinate of $\m$, and so the limits (\ref{pres}) and
 (\ref{prest}) below exist. The quantities
 \begin{align}
  P_\Gam(\u) &:= \lim_{\m \to \infty} \frac{\log Z_\Gam(\m,\u)}{\vol(\m)}, \label{pres} \\
  P_{\Gam,\topo}(\u) &:= \lim_{\m \to \infty}
  \frac{\log Z_{\Gam,\topo}(\m,\u)}{\vol(\m)},
  \label{prest} \\
  P_{\Gam,\perio}(\u) &:= \limsup_{\m \to \infty} \frac{\log Z_{\Gam,\perio}(\m,\u)}
  {\vol(\m)}
  \label{perent}
 \end{align}
 are called the \emph{pressure}, the \emph{topological pressure}
 and the \emph{periodic pressure} of $C_\Gam(\Z^d)$,
 respectively. The special cases
 $h_\Gam := P_\Gam(\0)$, $h_{\Gam,\topo} := P_{\Gam,\topo}(\0)$,
 $h_{\Gam,\perio} := P_{\Gam,\perio}(\0)$ are the entropy, the
 topological entropy and the periodic entropy, respectively,
 discussed in \cite{FP}. Clearly
 \[-\infty \leq P_{\Gam,\perio}(\u) \leq  P_{\Gam,\topo}(\u) \leq P_\Gam(\u).\]
 By the log-convexity of $Z_\Gam(\m,\u)$, $Z_{\Gam,\topo}(\m,\u)$,
 $Z_{\Gam,\perio}(\m,\u)$, it follows that $P_\Gam(\u)$, $P_{\Gam,\topo}(\u)$,
 $P_{\Gam,\perio}(\u)$  are convex functions on $\R^n$. (We
 agree that the constant function $-\infty$ is convex.) As in
 \cite{Fr1}, one has the equality
 \[P_{\Gam,\topo}(\u) = P_\Gam(\u).\]
 Since $\log Z_\Gam(\m,\u)$ is subadditive in each coordinate of $\m$, it
 follows that
 \begin{equation}\label{ub1}
 P_\Gam(\u) \leq \frac{\log Z_\Gam(\m,\u)} {\vol(\m)}.
 \end{equation}

 In the one-dimensional case $d=1$, we can express $P_{\Gamma_1}(\u)$ as the
 logarithm of the spectral radius (largest modulus of an
 eigenvalue) of a certain $n \times n$ matrix as follows.
 \begin{prop}\label{onedimpressure}
  Let $D_{\Gamma_1} = (d_{ij})_{i,j \in \an{n}}$ be the
  $(0,1)$-adjacency matrix of $\Gamma_1$, and let $D_{\Gamma_1}(\u) =
  (d_{ij}(\u))_{i,j \in \an{n}}$ be defined by
  \begin{equation}\label{AGamma1u}
  d_{ij}(\u) := d_{ij} \cdot  e^{\frac{1}{2}(\e_i\trans \u + \e_j\trans \u)}.
  \end{equation}
  Let $\rho_{\Gamma_1}(D(\u))$ be the spectral radius of
  $D_{\Gamma_1}(\u)$.
  Then
  \[P_{\Gamma_1}(\u) = \log \rho_{\Gamma_1}(D(\u)).\]
 \end{prop}
 \begin{proof}
  Recall the following characterization of the spectral radius
  of a nonnegative matrix $M$: for any vector $\w$ with positive
  components, we have $\rho(M) = \lim_{k \to \infty} (\w\trans M^k
  \w)^{\frac{1}{k}}$ (cf., Proposition 10.1 of \cite{Fr2}).
  Consider the positive vector
  \[\1(\u) = (e^{\frac{1}{2}\e_i\trans \u})_{i \in \an{n}}.\]
  Since $C_{\Gamma_1}(\an{m_1})$ is the set of walks of length $m_1 - 1$ on
  $\Gamma_1$, we have $Z_{\Gamma_1}(\an{m_1},\u) = \1(\u)\trans D_{\Gamma_1}(\u)^{m_1 -1}
  \1(\u)$. Therefore
  \begin{multline*}
   \log \rho_{\Gamma_1}(D(\u)) = \lim_{m_1 \to \infty} \frac{\log \1(\u)\trans D_{\Gamma_1}(\u)^{m_1} \1(\u)}{m_1}
   \\
   = \lim_{m_1 \to \infty} \frac{\log \1(\u)\trans D_{\Gamma_1}(\u)^{m_1 - 1} \1(\u)}{m_1}
   = \lim_{m_1 \to \infty} \frac{\log Z_{\Gamma_1}(m_1,\u)}{m_1}
   = P_{\Gamma_1}(\u)
  \end{multline*}
  \end{proof}

 \section{Main Inequalities for Symmetric NNSOFT}\label{sec:SymmetricNNSOFT}

 In this section we derive bounds for the pressure analogous
 to those for the entropy in \cite[Section 3]{FP} under the assumption
 that some of the digraphs $\Gamma_1,\ldots,\Gamma_d$ are symmetric.

 For $d \geq 2$, consider $\m = (m_1,\ldots,m_d) \in \N^d$ and
 $\m^{-} := (m_2,\ldots,m_d)$. We denote by $T(\m)$ the discrete
 torus with sides of length $m_1,\ldots,m_d$, i.e., direct product
 of cycles of lengths $m_1,\ldots,m_d$. Let
 $C_{\Gam,\perio,\{1\}}(\m)$ be the set of $\Gam$-colorings of the
 box $\an{\m}$ that correspond to $\Gam$-colorings of $T(m_1) \times
 \an{\m^{-}}$, i.e., that can be extended periodically in the
 direction of $\e_1$ with period $m_1$ into $\Gam$-colorings of $\Z
 \times \an{\m^{-}}$. We can view these colorings as
 $\widehat{\Gam}$-colorings of the box $\an{\m^{-}}$, where
 $\widehat{\Gam} = (\widehat{\Gamma}_2,\ldots,\widehat{\Gamma}_d)$,
 for each $k$ the vertex set of $\widehat{\Gamma}_k$ is the set
 $\Gamma_{1,\perio}^{m_1}$ of closed walks $a =
 (a_1,\ldots,a_{m_1},a_1)$ of length $m_1$ on $\Gamma_1$, and $(a,b)
 \in \widehat{\Gamma}_k$ if and only if $(a_i,b_i) \in \Gamma_k$ for
 $i =1,\ldots,m_1$. For this reason, the limit (\ref{defbarh}) below
 exists and is equal to the pressure $P_{\widehat{\Gam}}(\u)$ of the
 NNSOFT $C_{\widehat{\Gam}}(\Z^{d-1})$:
 \begin{align}
 Z_{\Gam,\perio,\{1\}}(\m,\u) &:= \sum_{\phi \in C_{\Gam,\perio,\{1\}}(\m)}
 e^{\bc(\phi)\trans \u},
 \nonumber\\
 \overline{P}_{\Gam}(m_1,\u) &:= \lim_{\m^- \to \infty} \frac{\log
 Z_{\Gam,\perio,\{1\}}(\m,\u)}{\vol(\m^{-})}, \quad m_1 \in \N.
 \label{defbarh}
 \end{align}
 Then $\overline{P}_{\Gam}(m_1,\u)$ is a convex function of
 $\u\in\R^n$. In the degenerate case $m_1 = 0$, we define
 $Z_{\Gam,\perio,\{1\}}((0,\m^{-}),\u)$
 to be $\#C_\Gam^{-}(\m^{-})$
 (regardless of $\u$), where $C_\Gam^{-}(\m^{-})$ is the set of
 $(\Gamma_2,\ldots,\Gamma_d)$-colorings of the box $\an{\m^{-}}$.
 Then (\ref{defbarh}) is also valid for $m_1=0$, where
 $\overline{P}_{\Gam}(0,\u) := P_{(\Gamma_2,\ldots ,\Gamma_d)}(\0)$
 is the entropy of $C_{(\Gamma_2,\ldots,\Gamma_d)}(\Z^{d-1})$.

 \begin{theo}\label{perulb}  Consider the NNSOFT
  $C_\Gam(\Z^d)$ for $d \geq 2$, and
  let $P_\Gam(\u)$ and $\overline{P}_{\Gam}(m_1,\u)$ be defined by (\ref{pres}) and
  (\ref{defbarh}), respectively.
  Assume that $\Gamma_1$ is symmetric.  Then for all $p,r \in \N$ and $q \in \Z_+$,
  \begin{equation}\label{perulb1}
   \frac{\overline{P}_{\Gam}(2r,\u)}{2r} \geq  P_\Gam(\u)  \geq
   \frac{\overline{P}_{\Gam}(p + 2q,\u)- \overline{P}_{\Gam}(2q,\u)}{p}.
  \end{equation}
 \end{theo}
 \begin{Proof}  Fix $\m^{-} = (m_2,\ldots,m_d) \in \N^{d-1}$ and let
 $\Omega_1(\m^{-})$ be the following transfer digraph on the
 vertex set $C_\Gam^{-}(\m^{-})$, analogous to the transfer digraph
 $\Omega_d(\m')$ described in \cite[Section 1]{FP}. Vertices $\i$ and
 $\j$ satisfy $(\i,\j) \in \Omega_1(\m^{-})$ if and only if
 $[\i,\j] \in C_\Gam(2,\m^{-})$, where $[\i,\j]$  is the configuration
 consisting of $\i,\j$ occupying the levels $x_1=1,2$ of
 $\an{(2,\m^{-})}$, respectively. Let $N = \# C_\Gam^{-}(\m^{-})$ and let
 $D_\Gam(\m^{-})=(d_{\i\j})_{\i,\j \in C_\Gam^{-}(\m^{-})}$ be the $N \times
 N$ $(0,1)$-adjacency matrix of $\Omega_1(\m^{-})$.  Let
 $D_\Gam(\m^{-},\u)=(d_{\i\j}(\u))_{\i,\j \in C_\Gam^{-}(\m^{-})}$ be
 defined by
 \begin{equation}\label{Amu}
 d_{\i\j}(\u) = d_{\i\j} \cdot e^{\frac{\bc(\i)\trans\u +\bc(\j)\trans\u}{2}},
 \quad \i,\j \in C_\Gam^{-}(\m^{-}),
 \end{equation}
 and let the positive vector $\1(\u)$ be defined by
 \[\1(\u)=(e^{\frac{\bc(\i)\trans\u}{2}})_{\i \in C_\Gam^{-}(\m^{-})}.\]
 Then
 \[\1(\u)\trans D_\Gam(\m^{-},\u)^{m_1} \1(\u) = Z_\Gam((m_1,\m^{-}),\u),\]
 and as in the proof of Proposition~\ref{onedimpressure}
 \[\log \rho(D_\Gam(\m^{-},\u)) = \lim_{m_1 \to \infty} \frac{\log \1(\u)\trans
 D_\Gam(\m^{-},\u)^{m_1} \1(\u)}{m_1}.\]
 (In particular, $\rho(D_\Gam(\m^{-},\u))$ is a log-convex
 function of $\u$ \cite{Kin}.)
 It follows that
 \begin{equation}\label{logrhoeqlim}
 \frac{\log\rho(D_\Gam(\m^{-},\u))} {\vol(\m^{-})} =
 \lim_{m_1 \to \infty} \frac{\log Z_\Gam((m_1,\m^{-}),\u)}{m_1\,\vol(\m^{-})}.
 \end{equation}
 Now send $m_2,\ldots,m_d$ to $\infty$, and observe that by
 (\ref{pres}) and (\ref{ub1}), the right-hand side of
 (\ref{logrhoeqlim}) converges to $P_\Gam(\u)$ and
 bounds it from above for each $\m^{-}$. Thus we obtain an analog of \cite{Fr1}
 \begin{align}
  \frac{\log\rho(D_\Gam(\m^{-},\u))}{\vol(\m^{-})} &\geq
  P_\Gam(\u), \qquad \m^{-} \in \N^{d-1}\label{enteq}\\
  \lim_{\m^- \to \infty}\frac{\log\rho(D_\Gam(\m^{-},\u))}
  {\vol(\m^{-})} &= P_\Gam(\u).\label{enteq1}
 \end{align}
 Next, we observe that
 \begin{equation}\label{percount}
 \tr D_\Gam(\m^{-},\u)^{q} = Z_{\Gam,\perio,\{1\}}((q,\m^{-}),\u), \qquad q \in \Z_+,
 \end{equation}
 where $D_\Gam(\m^{-},\u)^0$ is the $N \times N$ identity matrix. Recall
 that
 \[ \tr D_\Gam(\m^{-},\u)^{q}=\sum_{i=1}^N \lambda_i^{q}, \qquad q \in \Z_+,\]
 where $\lambda_1,\ldots,\lambda_N$ be the eigenvalues of
 $D_\Gam(\m^{-},\u)$. Since $D_\Gam(\m^{-},\u)$ is a nonnegative
 matrix, its spectral radius $\rho(D_\Gam(\m^{-},\u)) := \max_{i \in
 \an{N}} |\lambda_i|$ is one of the $\lambda_i$ by the
 Perron-Frobenius theorem. Since by assumption $\Gamma_1$ is
 symmetric, $\Omega_1(\m^{-})$ and hence $D_\Gam(\m^{-},\u)$ are
 symmetric. Therefore $\lambda_1, \ldots,\lambda_N$ are real, and
 hence $\tr D_\Gam(\m^{-},\u)^{2r} \geq
 \rho(D_\Gam(\m^{-},\u))^{2r}$ for each $r \in \N$. Taking
 logarithms and using (\ref{percount}), we obtain
 \begin{equation}\label{basineq}
 \frac{\log Z_{\Gam,\perio,\{1\}}((2r,\m^{-}),\u)}{2r \vol(\m^{-})} \geq
 \frac{\log\rho(D_\Gam(\m^{-},\u))}{\vol(\m^{-})}, \qquad r \in \N.
 \end{equation}
 Sending $m_2,\ldots,m_d$ to $\infty$ in (\ref{basineq}) and using
 (\ref{defbarh}) and (\ref{enteq1}), we
 deduce the upper bound for $P_\Gam(\u)$ in (\ref{perulb1}).

 To prove the lower bound in (\ref{perulb1}), we note that
 \begin{multline*}
 \tr D_\Gam(\m^{-},\u)^{p+2q} = \sum_i \lambda_i^{p+2q} \leq \sum_i |\lambda_i|^{p+2q}
 = \sum_i |\lambda_i|^p \lambda_i^{2q}\\
 \leq \sum_i \rho(D_\Gam(\m^{-},\u))^p \lambda_i^{2q} = \rho(D_\Gam(\m^{-},\u))^p
 \tr D_\Gam(\m^{-},\u)^{2q}
 \end{multline*}
 and thus by (\ref{percount})
 \begin{equation}
  \rho(D_\Gam(\m^{-}),\u)^p  \geq  \frac{\tr D_\Gam(\m^{-},\u)^{p+2q}}{\tr
  D_\Gam(\m^{-},\u)^{2q}}=\frac{Z_{\Gam,\perio,\{1\}}((p+2q,\m^{-}),\u)}
  {Z_{\Gam,\perio,\{1\}}((2q,\m^{-}),\u)}.\label{rhogeqtraceratio}
 \end{equation}
 Therefore
 \[ \frac{\log\rho(D_\Gam(\m^{-},\u))}{\vol(\m^{-})}  \geq
  \frac{\log Z_{\Gam,\perio,\{1\}}((p+2q,\m^{-}),\u) -
  \log Z_{\Gam,\perio,\{1\}}((2q,\m^{-}),\u)}{p\,\vol(\m^{-})}.
 \]
 Sending $\m^-$ to $\infty$ and using (\ref{enteq1}) and
 (\ref{defbarh}) (recall that the latter holds for $m_1 \in \Z_+$), we
 deduce the lower bound in (\ref{perulb1}).  \end{Proof}

 When $d=2$, $\overline{P}_{\Gam}(m_1,\u)$ is the pressure of the
 NNSOFT $C_{\widehat{\Gamma}_2}(\Z)$ (recall that
 $\overline{P}_{\Gam}(0,\u)$ is the entropy $h_{\Gamma_2}$). Since
 this is a $1$-dimensional NNSOFT,
 Proposition~\ref{onedimpressure} implies
 that $\overline{P}_{\Gam}(m_1,\u) =
 \log\rho(D_{\widehat{\Gamma}_2}(\u))$, where $D_{\widehat{\Gamma}_2}(\u)$
 is defined as in (\ref{AGamma1u}). We denote
 $\rho(D_{\widehat{\Gamma}_2}(\u))$ by $\theta_2(m_1,\u)$, and
 obtain the following corollary to Theorem~\ref{perulb}.

 \begin{corol}\label{ulbd=2}
 Let $d=2$ and assume that $\Gamma_1$ is symmetric.
 Then for all $p,r \in \N$ and $q \in \Z_+$,
 \begin{equation}\label{ubspecr2}
 \frac{\log\theta_2(2r,\u)}{2r} \geq  P_\Gam(\u) \geq
 \frac{\log\theta_2(p+2q,\u)-\log\theta_2(2q,\u)}{p},
 \end{equation}
 where $\theta_2$ is defined above.
 \qed
 \end{corol}

 In (\ref{ubspecr2}) take $q=0$ and $p=2r$, and send $r$ to
 $\infty$.  Clearly the upper and lower bounds then converge to
 $P_\Gam(\u)$. Hence $P_\Gam(\u)$ is computable, as shown in \cite{Fr1}
 for the entropy $P_{\Gam}(\0)$.
 Combining the arguments of the proof of Theorem \ref{perulb} with
 the arguments of the proof of Theorem 3.4 in \cite{FP}, we obtain
 \begin{equation}\label{entdef}
 P_\Gam(\u) \leq \frac{\log\rho(D_\Gam(\m^{-},\u))}{\vol(\m^-)},
 \qquad m_2,\ldots ,m_{d} \text{ even, }\;
 \Gamma_2,\ldots,\Gamma_{d} \text{ symmetric},
 \end{equation}
 where  $D_\Gam(\m^{-},\u)$ is defined in (\ref{Amu}).

  \section{The Conjugate of Pressure and the Density
  Entropy}\label{densityentropy}

  The purpose of this section to exhibit a striking connection
  between the conjugate function $P_\Gam^*(\p)$ of the
  pressure $P_\Gam(\u)$ and the density entropy $h_{\Gamma}(\p)$
  for certain probability vectors $\p$.

  First we need to recall some properties of convex functions,
  which can be found in \cite{Roc}.  We adopt the notations of that
  book.

  In this paper we consider only convex functions $f:\R^m \to \R
  \cup \set{+\infty}$ that are not identically equal to $+\infty$.
  Such convex functions are called \emph{proper} in \cite{Roc}. Let
  $f$ be a proper convex function.  Then $\dom f:=\{\x \in \R^m :\;
  f(\x) < \infty\}$, which is called the \emph{effective domain} of
  $f$, is a nonempty convex set in $\R^m$. Let $L \subseteq \R^m$ be
  the minimal affine subspace that contains $\dom f$, and let $l$ be
  its dimension. The affine transformation $A$ that maps $L$ onto
  $\R^l$ maps $\dom f$ onto a convex set $C \subseteq \R^l$. We
  denote the interior of $C$ by $\inter C$. Then $\ri (\dom f) :=
  A^{-1}(\inter C)$ is called the \emph{relative interior} of $\dom f$
  (note that if $l=0$, then $\dom f$ and $\ri (\dom f)$ consist of
  the same single point). A proper convex function $f$ is
  Lipschitzian relative to any closed bounded subset of $\ri (\dom
  f)$ \cite[Thm 10.4]{Roc}.  In particular, $f$ is continuous on
  $\ri (\dom f)$ \cite[Thm 10.1]{Roc}.

  A proper convex function $f$ is called \emph{closed} if $f$ is
  lower semi-continuous \cite[Section~7, p.\ 52]{Roc}. In particular, if
  $f:\R^m \to \R$ is convex, then $\dom f = \R^m = \ri(\dom f$), $f$ is a
  continuous function on $\R^m$, hence closed, and $f$ is
  Lipschitzian relative to any closed bounded subset of $\R^m$.


  A vector $\y \in \R^m$ is called a \emph{subgradient of $f$ at $\x
  \in \R^m$} if $f(\z) \geq f(\x) + \y\trans (\z-\x)$ for all $\z
  \in \R^m$.  The set of all subgradients $\y$ at $\x$ is called
  \emph{the subdifferential of $f$ at $\x$} and is denoted by
  $\partial f(\x)$. As usual, for any set $S \subseteq \R^m$,
  $\partial f(S)$ denotes $\cup_{x \in S} \partial f(\x)$. Obviously
  $\partial f(\x)$ is a closed convex set.  If $\partial f(\x) \neq
  \emptyset$, then $f$ is said to be \emph{subdifferentiable at
  $\x$}.  A proper convex function $f$ is not subdifferentiable at
  any $\x \notin \dom f$, but is subdifferentiable at each $\x \in
  \ri (\dom f)$ \cite[Thm 23.4]{Roc}. Recall that $f$ is
  \emph{differentiable at $\x$} if there exists a vector $\nabla
  f(\x) = \y \in \R^m$ (necessarily unique) such that $f(\x + \w) =
  f(\x) + \y\trans \w + o(\|\w\|)$, $\w \to \0$. The vector $\nabla
  f(\x)$ is called \emph{the gradient of $f$ at $\x$}. We denote by
  $\Der f$ the set of all points where $f$ is
  differentiable, so $\partial f(\Der f)$ denotes the set of all gradient
  vectors of $f$.
  A proper convex function $f$ is differentiable at
  a point $\x \in \dom f$ if and only if $\partial f(\x)$ consists
  of a single point, which is then $\nabla f(\x)$ \cite[Thm
  25.1]{Roc}.

  Assume that $f$ is a proper convex function and $\inter(\dom f)
  \neq \emptyset$. Then $\Der f$ is a dense subset of $\inter (\dom
  f)$, $f$ is differentiable a.e.\ (almost everywhere) in $\inter
  (\dom f)$, and $\nabla f$ is continuous on $\Der f$ \cite[Thm
  25.5]{Roc}.
  Moreover, for each $\x\in \inter (\dom f) \setminus \Der f$, the
  convex set $\partial f(\x)$, which consists of more than one
  point, can be reconstructed as follows from the values of the
  gradient function $\nabla f$ on $\Der f$. Let $S(\x)$ consist of
  all the limits of sequences $\nabla f(\x_i)$ such that $\x_i \in
  \Der f$ and $\x_i \to \x$. Then $S(\x)$ is a closed and bounded
  subset of $\R^m$ and $\partial f (\x)=\conv S(\x)$ \cite[Thm 25.6,
  7.4]{Roc}.

  We now recall properties of the conjugate of a convex function $f$
   \cite[Section~12]{Roc}, denoted by $f^*$:
  \[f^*(\y) := \sup_{\x \in \R^n} \x\trans \y - f(\x) \quad \text{for each }
  \y \in \R^m.\]
  Since we assumed that $f$ is proper, it follows that $f^*$ is a
  proper closed convex function; moreover, if
  $f$ is closed then $f^{**}=f$ \cite[Thm. 12.2]{Roc}.
  A straightforward argument shows that
  \begin{equation}\label{fstari}
  f^*(\y)=\x\trans \y - f(\x) \quad \text{ for each subgradient } \y \in \partial f(\x).
  \end{equation}

  Recall that if $f$ is closed then $\partial f^*$ is the inverse
  of $\partial f$ in the sense of multivalued mappings, i.e.,
  $\x \in \partial f^*(\y)$ if and only if $\y \in \partial f(\x)$ \cite[Cor
  23.5.1]{Roc}.  In what follows we need the following result:

  \begin{lemma}\label{ransubf}
  Let $f$ be a proper closed convex function on $\R^m$. Then
  $\partial f(\R^m)$ is exactly the set of points in $\R^m$ where
  $f^*$ is subdifferentiable. In particular, $\ri (\dom f^*)
  \subseteq \partial f(\R^m) \subseteq \dom f^*$, and the closure of
  $\partial f(\R^m)$ is equal to the closure of $\dom f^*$.
  \end{lemma}
  \begin{Proof} Assume that $f^*$ is subdifferentiable at $\y$. Then there
  exists some $\x \in \R^m$ such that $\x \in \partial f^*(\y)$.
  Since $f$ is closed, it now follows that $\y \in \partial f(\x)$,
  so $\y \in \partial f(\R^m)$. The converse is shown in the same
  way. The first statement of the ``in particular'' now follows
  since $f^*$ is proper and hence, as noted above, is
  subdifferentiable in $\ri (\dom f^*)$ but not outside $\dom f^*$.
  The second statement of the ``in particular'' follows from the
  first one and from the fact that for a convex set $S$ such as
  $\dom f^*$, $\ri S$ and $S$ have the same closure \cite[Thm
  6.3]{Roc}. \end{Proof}

  We return to a general NNSOFT $C_\Gam(\Z^d)$. We assume throughout
  that $C_\Gam(\Z^d) \neq \emptyset$, for otherwise the pressure
  function $P_\Gam$ is identically $-\infty$.
  \begin{prop}\label{conthgu}
  The pressure $P_\Gam$ is a convex nonexpansive Lipschitz function
  on $\R^n$ with constant at most $1$ with respect to the norm
  $\|(v_1,\ldots,v_n)\|_{\infty} := \max_{i \in \an{n}} |v_i|$ :
  \begin{equation}\label{lipcon}
   |P_\Gam(\u+\v) - P_\Gam(\u)| \leq \|\v\|_{\infty}, \quad \u,\v \in \R^n.
  \end{equation}
  In particular, $P_\Gam$ is finite throughout $\R^n$ (this also follows from
  (\ref{ub1})). Therefore it is a proper closed convex function.
  \end{prop}
  \begin{Proof}  The convexity of $P_\Gam$ was pointed out in Section~1.
  Let $\phi \in C_\Gam(\an{\m})$.  Then
  \[|\bc(\phi)\trans \v| \leq \vol(\m) \|\v\|_{\infty}.\]
  Therefore a term-by-term comparison gives
  \[e^{-\vol(\m)\|\v\|_{\infty}} Z_\Gam(\m,\u) \leq Z_\Gam(\m,\u+\v) \leq
  e^{\vol(\m)\|\v\|_{\infty}} Z_\Gam(\m,\u).\]
  Take logarithms and divide by $\vol(\m)$ to obtain
  \begin{equation}\label{liplogw}
  \left| \frac{\log Z_\Gam(\m,\u + \v)}{\vol(\m)} -
         \frac{\log Z_\Gam(\m,\u)}{\vol(\m)}\right|
         \leq \|\v\|_{\infty}.
  \end{equation}
  Letting $\m \to \infty$, we deduce (\ref{lipcon}).  \end{Proof}

  Since $P_\Gam:\R^n\to \R$ is a convex function,
  it is differentiable almost everywhere.
  Consider the set of probability distributions on the set of
  colors
  \[\Pi_n := \{\p = (p_1,\ldots,p_n) : p_1,\ldots,p_n
  \geq 0,\; p_1+\cdots+p_n = 1\}.\]
  For
  $m \in \N$, we denote
  \[\Pi_n(m):=\{\bc=(c_1,\ldots,c_n) \in \Z_+^n :  c_1 + \cdots + c_n = m\} = m\Pi_n \cap \Z^n.\]
  Let
  \[C_\Gam(\an{\m},\bc):=\{\phi \in C_\Gam(\an{\m}) : \bc(\phi)=\bc\}, \quad \text{for all }
  \bc \in \Pi_n(\vol(\m)).\]
  This is the set of $\Gam$-colorings of $\an{\m}$ with color frequency vector $\bc$.
  \begin{defn}\label{deflamb1}
   A probability distribution $\p \in \Pi_n$ is called a density point of $C_\Gam(\Z^d)$
   when there exist sequences of boxes $\an{\m_q} \subseteq \N^d$ and
   color frequency vectors $\bc_q \in \Pi_n(\vol(\m_q))$ such that
   \begin{equation}\label{deflamb2}
    \m_q \to \infty, \;\; C_\Gam(\an{\m_q},\bc_q) \neq \emptyset \;\ \forall q \in \N, \;\;
    \text{ and } \lim_{q \to \infty} \frac{\bc_q}{\vol(\m_q)} = \p.
   \end{equation}
   We denote by $\Pi_{\Gam}$ the set of all density points of $C_\Gam(\Z^d)$.
   For $\p \in \Pi_{\Gam}$ we let
   \begin{equation}\label{hGamLam}
   h_{\Gam}(\p) := \sup_{\m_q,\bc_q} \limsup_{q \to \infty}
   \frac{\log\#C_\Gam(\an{\m_q},\bc_q)}{\vol(\m_q)}\ge 0,
   \end{equation}
   where the supremum is taken over all sequences satisfying
   (\ref{deflamb2}). One can think of $h_{\Gam}(\p)$ as the entropy
   for color density $\p$, called here the density entropy.
  \end{defn}
  It is straightforward to show (using a variant of the Cantor
  diagonal argument) that $\Pi_{\Gam}$ is a closed set.
  Furthermore, $h_{\Gam}$ is upper semi-continuous on $\Pi_{\Gam}$,
  because it is defined as a supremum.

  \begin{theo}\label{maxchar}
   Let $P_\Gam^*$ be
   the conjugate convex function
   of the pressure function $P_\Gam$.
   Then
   { 
   \renewcommand{\theenumi}{\alph{enumi}}
   \renewcommand{\labelenumi}{(\theenumi)}
   \begin{enumerate}
     \item \label{two} $h_\Gam(\p) \leq -P^*_\Gam(\p)$ for all $\p \in
     \Pi_\Gam$.
     \item \label{four}
      \begin{equation}\label{maxchar1}
      P_\Gam(\u) = \max_{\p \in \Pi_{\Gam}}
      (\p\trans\u + h_{\Gam}(\p)) \quad \text{ for all } \u\ \in
      \R^n.
      \end{equation}
      For $\u \in \R^n$, we denote
     \begin{equation}\label{defpu}
      \bpi(\u) := \arg \max_{\p \in \Pi_\Gam} (\p\trans \u + h_{\Gam}(\p))
      = \{\p \in \Pi_\Gam : P_\Gam(\u) = \p\trans \u + h_{\Gam}(\p)\},
     \end{equation}
     \item \label{six} For each $\p \in \bpi(\u)$, $h_\Gam(\p)= -P^*_\Gam(\p)$.
     \item \label{five} $\bpi(\u)\subseteq \partial P_{\Gam}(\u)$. In
     particular, if $\u \in \Der P_\Gam$, then $\bpi(\u) = \{\nabla
     P_\Gam(\u)\}$. Therefore $\partial P_\Gam(\Der P_\Gam) \subseteq
     \Pi_\Gam$.
     \item \label{threebis}
     Let $\u \in \R^n \setminus \Der P_\Gam$, and let $S(\u)$
     consist of all the limits of sequences $\nabla P_\Gam(\u_i)$
     such that $\u_i \in \Der P_\Gam$ and $\u_i \to \u$. Then $S(\u)
     \subseteq \bpi(\u)$.
     \item \label{one}
     $\conv \bpi(\u)=\conv S(\u)=\partial P_\Gam(\u)$.  Hence
     $\partial P_\Gam(\R^n)\subseteq \conv \Pi_\Gam \subseteq \Pi_n$.
     \item \label{onebis} $\conv \Pi_\Gam = \dom P^*_\Gam$.
   \end{enumerate}
   } 
  \end{theo}
  \begin{Proof}  First we show that $P_\Gam(\u) \geq \p\trans \u +
  h_{\Gam}(\p)$ for all $\p \in \Pi_{\Gam}$.  Fix $\p \in
  \Pi_{\Gam}$ and let $\m_q,\bc_q$, $q \in \N$, be sequences
  satisfying (\ref{deflamb2}). We have $Z_\Gam(\m_q,\u) \geq
  \#C_\Gam(\m_q,\bc_q) e^{\bc_q\trans \u}$, since the right-hand side
  is just one term of the sum represented by left-hand side. Take
  logarithms, divide by $\vol(\m_q)$, take $\limsup_{q \to
  \infty}$ and use the definition of $P_\Gam(\u)$ and the limit in
  (\ref{deflamb2}) to deduce $P_\Gam(\u) \geq \p\trans \u +
  \limsup_{q \to \infty} \frac{\log\#C_\Gam(\m_q,\bc_q)}{\vol(\m_q)}$.
  Now take the supremum over all sequences $\m_q, \bc_q$ satisfying
  (\ref{deflamb2}) and use (\ref{hGamLam}) to obtain
  \begin{equation}\label{upinhgu1}
   P_\Gam(\u) \geq \p\trans \u + h_{\Gam}(\p) \quad \text{ for all } \p \in \Pi_\Gam
  \end{equation}
  and thus
  \begin{equation}\label{upinhgu}
  P_\Gam(\u) \geq \sup_{\p \in \Pi_{\Gam}}\p\trans \u + h_{\Gam}(\p).
  \end{equation}
  On the other hand, (\ref{upinhgu1}) can be written as $\p\trans \u
  - P_\Gam(\u) \leq -h_{\Gam}(\p)$ for all $\p \in \Pi_\Gam$, and
  then taking the supremum over $\u$ gives $P^*_{\Gam}(\p) \leq
  -h_{\Gam}(\p) < \infty $ for all $\p \in \Pi_\Gam$.  Thus we have
  established (\ref{two}) as well as
  \begin{equation}\label{pigammasubsetconjugate}
  \Pi_\Gam \subseteq \dom P^*_\Gam.
  \end{equation}
  We now show that for each $\u \in \R^n$ there exists $\p(\u) \in
  \Pi_{\Gam}$ satisfying $P_\Gam(\u) \leq \p(\u)\trans \u +
  h_{\Gam}(\p(\u))$, which together with (\ref{upinhgu}) will
  establish (\ref{maxchar1}).
  Observe first that
  \[\#\Pi_n(m)= \binom{m+n-1}{n-1} = O(m^{n-1}), \quad m \to \infty.\]
  Therefore for each $\m \in \N^d$,
  \[Z_\Gam(\m,\u) = O(\vol(\m)^{n-1}) \max_{\bc \in \Pi_n(\vol(\m))}
  \#C_\Gam(\an{\m},\bc) e^{\bc\trans \u}.\]
  Let
  \begin{equation}\label{defcmu}
  \bC(\m,\u) := \arg\max_{\bc \in \Pi_n(\vol(\m))}
  \#C_\Gam(\an{\m},\bc) e^{\bc\trans \u}.
  \end{equation}
  Then for $\bc(\m,\u) \in \bC(\m,\u)$ we have
  \begin{equation}\label{estimate}
  Z_\Gam(\m,\u) = O(\vol(\m)^{n-1}) \#C_\Gam(\an{\m},\bc(\m,\u)) e^{\bc(\m,\u)\trans \u}.
  \end{equation}
  Since $C_\Gam(\Z^d) \neq \emptyset$, for each $\m \in \N^d$ and
  $\u \in \R^n$, $\frac{\bc(\m,\u)}{\vol(\m)}$ is a well-defined point
  in $\Pi_n$.  Choose a sequence $\m_{q} \to \infty$ such that
  $\frac{\bc(\m_{q},\u)}{\vol(\m_{q})}$ converges to some $\p(\u)$.
  We have $\p(\u) \in \Pi_{\Gam}$ by Definition~\ref{deflamb1}. Apply
  (\ref{estimate}) to $\m_{q}$, and use the definition of
  $P_\Gam(\u)$ and $h_{\Gam}(\p(\u))$ to deduce
  \[P_\Gam(\u) \leq {\p(\u)}\trans \u + \limsup_{q \to \infty}
  \frac{\log\#C_\Gam(\an{\m_{q}},\bc(\m_{q},\u))}{\vol(\m_{q})} \leq
  {\p(\u)}\trans \u + h_{\Gam}(\p(\u)),\]
  which is the desired inequality establishing (\ref{four}).

  By the definition of $\bpi(\u)$, for each $\p \in \bpi(\u)$ we have
  $-h_{\Gam}(\p)=\p\trans \u - P_{\Gam}(\u) \leq  P^*_{\Gam}(\p)$.
  Combining this with (\ref{two}),
  we deduce (\ref{six}).

  Let $\p \in \bpi(\u)$ and $\v\in\R^n$.  Then the maximal
  characterization (\ref{maxchar1}) of $P_{\Gam}(\u + \v)$ and (\ref{defpu})
  give
  \[P_{\Gam}(\u + \v)\geq \p\trans
  (\u + \v) + h_{\Gam}(\p) = \p\trans \v + P_{\Gam}(\u).\]
  This proves (\ref{five}).


  We now prove (\ref{threebis}). Assume that $\u \in \R^n \setminus
  \Der P_\Gam$ and $\p \in S(\u)$. Then there exists a sequence
  $\u_i \in \Der P_\Gam$ such that $\u_i \to \u$ and $\nabla
  P_\Gam(\u_i) \to \p$. We have $\{\nabla P_\Gam(\u_i)\} = \bpi(\u_i)
  \subseteq
  \Pi_\Gam$ by (\ref{five}), and since $\Pi_\Gam$ is closed, $\p \in
  \Pi_\Gam$.
  By definition of $\bpi(\u_i)$ we have $P_\Gam(\u_i) = \nabla
  P_\Gam(\u_i)\trans \u_i + h_\Gam(\nabla P_\Gam(\u_i))$. When $i
  \to \infty$ we have firstly $P_\Gam(\u_i) \to P_\Gam(\u)$ by the
  continuity of $P_\Gam$, secondly $\nabla P_\Gam(\u_i)\trans \u_i
  \to \p\trans \u$, and thirdly $\limsup h_\Gam(\nabla P_\Gam(\u_i)
  \leq h_\Gam(\p)$ by the upper semi-continuity of $h_\Gam$.
  Therefore $P_\Gam(\u) \leq \p\trans \u + h_\Gam(\p)$. This, the
  fact that $\p \in \Pi_\Gam$, and (\ref{maxchar1}) show that
  $P_\Gam(\u) = \p\trans \u + h_\Gam(\p)$, which by definition means
  $\p \in \bpi(\u)$.

  We show the first first identity of (\ref{one}).
  Let
  $\u \in \R^n \setminus \Der P_\Gam$, and let $S(\u)$ be as in
  (\ref{threebis}). By (\ref{threebis}) we have $S(\u)\subseteq
  \Pi_\Gam(\u)$, and therefore $\partial P_\Gam(\u) = \conv S(\u)
  \subseteq \conv \Pi_\Gam(\u)$.  Since $\partial P_\Gam(\u)$ is
  convex, from the first claim of (\ref{five})
  we obtain $\partial P_\Gam(\u)
  \supseteq \conv \Pi_\Gam(\u)$.  Hence $\partial P_\Gam(\u)
  = \conv \Pi_\Gam(\u)$.  Clearly, $\conv \Pi_\Gam(\u) \subseteq
  \conv \Pi_\Gam$.  Hence
  $\partial P_\Gam(\R^n) \subseteq \conv \Pi_\Gam$. The second inclusion of
  the second claim of (\ref{one}) follows from $\Pi_\Gam \subseteq \Pi_n$, which holds
  by definition of $\Pi_\Gam$.

  We finally show (\ref{onebis}). By (\ref{pigammasubsetconjugate})
  and the convexity of $\dom P^*_\Gam$, we have $\conv \Pi_\Gam
  \subseteq \dom P^*_\Gam$. It is left to show that $\dom P^*_\Gam
  \subseteq \conv \Pi_\Gam$. By Lemma~\ref{ransubf} $\partial
  P_\Gam(\R^n)$ is the set of all points where $P^*_\Gam$ is
  subdifferentiable. In particular $\ri (\dom P^*_\Gam) \subseteq
  \partial P_\Gam(\R^n)$, so by (\ref{one}) we have $\ri (\dom
  P^*_\Gam) \subseteq \conv \Pi_\Gam$. Apply the closure operator to
  both sides of this inclusion. On the left we get $\cl (\dom
  P^*_\Gam)$ by the convexity of $\dom P^*_\Gam$ (by \cite[Thm
  6.3]{Roc}, every convex set $C$ satisfies $\cl(\ri C) = \cl C$).
  On the right we get $\conv \Pi_\Gam$ because $\Pi_\Gam$ is closed.
  So we obtain $\dom P^*_\Gam \subseteq \cl (\dom P^*_\Gam)
  \subseteq \conv \Pi_\Gam$, as required.
  \end{Proof}

  As $h_{\Gam} = P_{\Gam}(\0)$, we obtain from (\ref{maxchar1})
  the following generalization of \cite[(4.12)]{FP}, which deals
  with the case of monomer-dimer entropy:
  \begin{corol}\label{maxcharhgu}
   \[h_{\Gam} = \max_{\p \in \Pi_{\Gam}} h_{\Gam}(\p).\]
  \qed
  \end{corol}

  For each  $\p \in \Pi_\Gam(\R^n) := \bigcup_{\u \in \R^n}
  \Pi_\Gam(\u)$ we have $h_\Gam(\p) = -P_\Gam^*(\p)$ by (\ref{six})
  of Theorem (\ref{maxchar}). Since $P^*_\Gam$ is a convex function,
  we obtain the following generalization of the result Hammersley
  \cite{Ha1}.

  \begin{corol}\label{concavh*}
  The function $h_\Gam(\cdot) : \Pi_\Gam \to \R_+$ is concave
  on every convex subset of $\Pi_\Gam(\R^n)$.
  \qed
  \end{corol}

  To obtain the exact generalization of the result of Hammersley
  that $\Pi_\Gam$ is convex and $h_\Gam(\cdot) : \Pi_\Gam \to \R_+$
  is a concave function on the entire $\Pi_\Gam$, we need additional
  assumptions on the digraph $\Gam$, which do hold for the $\Gam$
  that codes the monomer-dimer tilings of $\Z^d$. For $\m \in \N^d$,
  if $\alpha : \an{\m} \to \an{n}$ is a coloring of a box $\an{\m}$
  and $\j \in \Z^d$, then to color the shifted box $\an{\m}+\j$ by
  $\alpha$ means to give to $\x+\j$ the color $\alpha(\x)$ for each
  $\x \in \an{\m}$. Recall that $C_\Gam(\an{\m})$ denotes the set of
  all $\Gam$-allowed colorings $\alpha : \an{\m} \to \an{n}$, that
  is to say, such that if $\x, \x+\e_i \in \an{\m}$, then
  $(\alpha(\x),\alpha(\y)) \in \Gamma_i$.

  \begin{defn}\label{dmrlgam}
   For a given digraph $\Gam = (\Gamma_1,\ldots, \Gamma_d)$ on the
   vertex set $\an{n}$, a set $\cF = \cup_{\m \in \N^d}
   \widetilde{C}_\Gam(\an{\m})$, where $\widetilde{C}_\Gam(\an{\m})
   \subseteq C_\Gam(\an{\m})$ for each $\m \in \N^d$, is called
   friendly if the following condition holds: whenever a shifted box
   is cut in two and each part is colored by a coloring in $\cF$,
   then the combined coloring also belongs to $\cF$. More precisely,
   let $\m, \n \in \N^d$ and $\j \in \Z^d$ be such that $\an{\m}
   \cap (\an{\n}+\j) = \emptyset$, and such that $T := \m \cup
   (\an{\n}+\j)$ is a box $\an{\k}+\i$ for some $\k \in \N^d$ and
   $\i \in \Z^d$. Let $\alpha \in \widetilde{C}_\Gam(\an{\m})$,
   $\beta \in \widetilde{C}_\Gam(\an{\n})$, and let $\gamma : T \to
   \an{n}$ color $\an{\m}$ by $\alpha$ and $\an{\n}+\j$ by $\beta$.
   Then the coloring $\delta : \an{\k} \to \an{n}$ defined by
   $\delta(\x) = \gamma(\x+\i)$ belongs to
   $\widetilde{C}_\Gam(\an{\k})$.

   The digraph $\Gam$ is called friendly if there exist a friendly
   set $\cF = \cup_{\m \in \N^d} \widetilde{C}_\Gam(\an{\m})$ and a
   constant vector $\b \in \N^d$ such that if any box $\an{\m}$ is
   padded with an envelope of width $b_i$ in the direction of
   $\e_i$, then each $\Gam$-allowed coloring of $\an{\m}$ can be
   extended in the padded part to a coloring in $\cF$. More
   precisely, for each $\m \in \N^d$ and each $\alpha \in
   C_\Gam(\an{\m})$, there exists a coloring in
   $\widetilde{C}_\Gam(\an{\m+2\b})$ that colors $\an{\m}+\b$ by
   $\alpha$.
  \end{defn}

  \begin{exam}\label{examfr}  Let $\Gam=(\Gamma_1,\ldots,\Gamma_d)$
  be a coloring digraph with vertex set $\an{n}$.  Then $\Gam$ is a
  friendly digraph with $b=1$ if one of the following conditions
  holds:
   { 
   \renewcommand{\theenumi}{\alph{enumi}}
   \renewcommand{\labelenumi}{(\theenumi)}
   \begin{enumerate}
    \item $\Gam$ has a friendly color $f\in\an{n}$,
    i.e., for each $i \in \an{d}$ we have $(f,j),(j,f) \in \Gamma_i$
    for all $j \in \an{n}$ (we can take $\widetilde{C}_\Gam(\m)$ to
    be those $\Gam$-allowed colorings of $\an{\m}$ whose boundary
    points are colored with $f$). This example is useful for the
    hard-core model with $n=2$ and $\Gamma_i = \set{(1,1),(1,2),(2,1)}$,
    $f=1$.
    \item $\Gam$ is the digraph associated with the
   monomer-dimer covering as defined in~(\ref{encoding}) (we can
   take $\widetilde{C}_\Gam(\an{\m})$ to be the set of tilings of
   $\an{\m}$ by monomers and dimers, i.e., the coverings in which no
   dimer protrudes out of $\an{\m}$, as in Hammersley).
  \end{enumerate}
   } 
  \end{exam}

  The following theorem strengthens Theorem~\ref{maxchar} and
  generalizes the results of Hammersley in case $\Gam$ is a friendly
  digraph.

 \begin{theo}\label{frnthm}  Let $\Gam = (\Gamma_1,\ldots,\Gamma_d)$
  be a friendly coloring digraph.  Then
  { 
   \renewcommand{\theenumi}{\alph{enumi}}
   \renewcommand{\labelenumi}{(\theenumi)}
  \begin{enumerate}
   \item\label{1frnd} $\Pi_\Gam$ is convex.  Hence $\Pi_\Gam = \dom
   P^*_\Gam$.
   \item\label{2frnd} $h_\Gam(\cdot) : \Pi_{\Gam} \to \R_+$ is concave.
   \item\label{3frnd} For each $\u \in \R^n$, $\Pi_\Gam(\u) = \partial
   P_\Gam(\u)$.
   \item \label{4frnd} For each $\u\in\R^n$,
   $h_\Gam(\cdot)$ is an affine function on $\partial P_\Gam(\u)$.
   \item \label{5frnd} $h_\Gam(\p)=-P^*_\Gam(\p)$ for each $\p\in\Pi_\Gam$.
  \end{enumerate}
   } 
 \end{theo}

 \begin{proof}
 We first first prove~(\ref{1frnd}).  Let $\alpha \in
 \widetilde{C}_\Gam(\an{\m})$, let $\bc(\alpha) = (c_1,\ldots,c_n)
 \in \Pi_n(\vol(\m))$ be the color frequency vector of $\alpha$, and
 let $\p := \frac{1}{\vol(\m)} \bc(\alpha)$.   We assert that $\p
 \in \Pi_\Gam$.  For $\k=(k_1,\ldots,k_d) \in \N^d$, we define $\k
 \cdot \m := (k_1m_1,\ldots,k_dm_d)$ and view $\an{\k\cdot\m}$ as a
 box composed of $\vol(\k)$ boxes isomorphic to $\an{\m}$, i.e., as
 $\an{\m}$ duplicated by a factor of $\k$.  We color each of these
 boxes by $\alpha$, obtaining a coloring $\alpha(\k \cdot \m)$ of
 $\an{\k \cdot \m}$. Clearly, $\p = \frac{1}{\vol(\k \cdot \m)}
 c(\alpha(\k \cdot \m))$. Since $\alpha \in
 \widetilde{C}_\Gam(\an{\m})$, it follows that $\alpha(\k \cdot \m)$
 belongs to $\widetilde{C}_\Gam(\k \cdot \m)$, so in particular is
 $\Gam$-allowed. Choosing a sequence $\k_q \to \infty$, we deduce
 that $\p = \lim_{q \to \infty} \frac{1}{\vol(\k_q \cdot \m)}
 c(\alpha(\k_q \cdot \m))$.  Hence $\p \in \Pi_\Gam$ according
 to~(\ref{deflamb2}), as asserted.

 Let $\beta \in \widetilde{C}_\Gam(\an{\n})$.  By the above argument
 we also have $\q := \frac{1}{\vol(\n)}\bc(\beta) \in \Pi_\Gam$. We
 assert that all $i,j \in \N$ satisfy $\frac{i}{i+j}
 \p+\frac{j}{i+j} \q \in \Pi_\Gam$. Let $\alpha(\n \cdot \m)$ and
 $\beta(\m \cdot \n)$ be defined as above. Notice that $\an{\n \cdot
 \m}$ is isomorphic to $\an{\m \cdot \n}$. By the above argument,
 $\alpha(\n \cdot \m) \in \widetilde{C}_\Gam(\n \cdot \m)$ and
 $\beta(\m \cdot \n) \in \widetilde{C}_\Gam(\m \cdot \n)$. We define
 $\k := (m_1n_1, \ldots, m_{d-1}n_{d-1},(i+j)m_dn_d)$ and view the
 box $\an{\k}$ as composed of $i+j$ boxes isomorphic to $\an{\m
 \cdot \n}$ aligned side-by-side along the direction of $\e_d$.
 Color the first $i$ of these boxes by $\alpha(\m \cdot \n)$ and the
 last $j$ by $\beta(\n \cdot \m)$, obtaining a coloring $\gamma$ of
 $\an{\k}$, which satisfies $\frac{1}{\vol(\k)} \bc(\gamma) =
 \frac{i}{i+j} \p + \frac{j}{i+j} \q$. Also $\gamma \in
 \widetilde{C}_\Gam(\an{\k})$, so in particular $\gamma$ is
 $\Gam$-allowed. By the above argument we obtain that $\frac{i}{i+j}
 \p + \frac{j}{i+j} \q \in \Pi_\Gam$, as asserted. Since $\Pi_\Gam$
 is closed we deduce that $a\p + (1-a)\q \in \Pi_\Gam$ for all $a \in
 [0,1]$.

 Let $\widetilde{\Pi}_\Gam$ be the convex hull of all points of the
 form $\frac{1}{\vol(\m)} \bc(\alpha)$ for some $\m$ and some
 $\alpha \in \widetilde{C}_\Gam(\an{\m})$. By the argument above we
 have $\widetilde{\Pi}_\Gam \subseteq \Pi_\Gam$. Let $\p \in
 \Pi_\Gam$. By Definition~\ref{deflamb1} there exist sequences $\m_q
 \to \infty$ and color frequency vectors $\bc_q \in
 \Pi_n(\vol(\m_q))$ satisfying~(\ref{deflamb2}). Let $\alpha_q \in
 C_\Gam(\an{\m_q},\bc_q)$. By Definition~\ref{dmrlgam}, $\alpha_q$
 can be extended to a coloring $\widetilde{\alpha}_q \in
 \widetilde{C}_\Gam(\an{\m_q + 2\b},\widetilde{\bc}_q)$ for some
 $\widetilde{\bc}_q$. Since $\b$ is constant and $\m_q \to \infty$,
 we have $\lim_{q \to \infty} \frac{1}{\vol(\m_q +
 2\b)}\widetilde{\bc}_q = \p$. Since $\frac{1}{\vol(\m_q +
 2\b)}\widetilde{\bc}_q \in \widetilde{\Pi}_\Gam$, we have $\p \in
 \cl \widetilde{\Pi}_{\Gam}$. Thus $\widetilde{\Pi}_\Gam \subseteq
 \Pi_\Gam \subseteq \cl \widetilde{\Pi}_{\Gam}$. Applying the
 closure operator, we deduce $\Pi_{\Gam} = \cl
 \widetilde{\Pi}_{\Gam}$, and since $\widetilde{\Pi}_{\Gam}$ is
 convex, so is $\Pi_\Gam$. The equality $\Pi_\Gam = \dom P^*_\Gam$
 follows from part~(\ref{onebis}) of Theorem~\ref{maxchar}.

 We now prove~(\ref{2frnd}).  Choose any $\varepsilon > 0$. Let $\p
 \in \Pi_\Gam$. By Definition~\ref{deflamb1} there exist sequences
 $\m_q \to \infty$ and color frequency vectors $\bc_q \in
 \Pi_n(\vol(\m_q))$ satisfying~(\ref{deflamb2}). By~(\ref{hGamLam})
 we may assume by selecting appropriate subsequences that the
 following limit exists and satisfies $\lim_{q \to \infty}
 \frac{\log \# C_\Gam(\an{\m_q},\bc_q)}{\vol(\m_q)} \geq h_\Gam(\p)
 - \varepsilon$. Each coloring $\alpha_q \in
 C_\Gam(\an{\m_q},\bc_q)$ can be extended to some coloring
 $\widetilde{\alpha}_q \in \widetilde{C}_\Gam(\an{\m_q  +2\b})$.
 Denote by $\bc(\widetilde{\alpha}_q)$ the color frequency vector of
 $\widetilde{\alpha}_q$. Since $b$ is constant and $\m_q \to
 \infty$, we have $ \lim_{q \to \infty} \frac{1}{\vol(\m_q+2\b)}
 \bc(\widetilde{\alpha}_q) = \p$.  Let $\cC_q \subseteq
 \Pi_n(\vol(\m_q + 2\b)$ be the set of all possible color frequency
 vectors of all extensions of the colorings of
 $C_\Gam(\an{\m_q},\bc_q)$ to $\widetilde{C}_\Gam(\an{\m_q +
 2\b})$. Clearly $\#\cC_q \leq \binom{\vol(\m_q+2\b) + n - 1}{n-1}
 = O(\vol(\m_q+2\b)^{n-1})$, $q \to \infty$. For each $\bc \in
 \Pi_n(\vol(\m))$, let $\widetilde{C}_\Gam(\an{\m},\bc)$ be the set
 of those colorings in $\widetilde{C}_\Gam(\an{\m})$ that have color
 frequency vector $\bc$. From the above it follows that
 \begin{equation*}
  \begin{split}
   \#C_\Gam(\an{\m_q},\bc_q) \leq
   \sum_{\bc \in \cC_q} \#\widetilde{C}_\Gam(\an{\m_q + 2\b},\bc) \leq
   \#\widetilde{C}_\Gam(\an{\m_q+2\b},\widetilde{\bc}_q) \#\cC_q
   \\
   = \#\widetilde{C}_\Gam(\an{\m_q+2\b},\widetilde{\bc}_q)O(\vol(\m_q+2\b)^{n-1})
  \end{split}
 \end{equation*}
 for some $\widetilde{\bc}_q \in \cC_q$. Taking logarithms, dividing
 by $\vol(\m_q + 2\b)$, and noting that

 \noindent
 $\lim_{q \to \infty}
 \frac{\vol(\m_q + 2\b)}{\vol(\m_q)} = 1$, we deduce that
 \[\lim_{q \to \infty} \frac{\widetilde{\bc}_q}{\vol(\m_q + 2\b)} = \p
 \quad\text{and}\quad
 \lim_{q \to \infty} \frac{\log \#\widetilde{C}_\Gam(\an{\m_q + 2\b},
 \widetilde{\bc}_q)}{\vol(\m_q + 2\b)}
 \geq h_\Gam(\p) - \varepsilon.\]

 Thus for $\p,\q\in \Pi_\Gam, \varepsilon>0$ we have sequences $\m_q:=(m_{1,q},\ldots,m_{d,q}),
 \n_q:=(n_{1,q},\ldots,n_{d,q})\in\N^d, q\in\N$, with
 $\m_q,\n_q\to\infty$ such that the following two conditions hold:
 \begin{eqnarray*}
 \widetilde{C}_\Gam(\an{\m_q},\bc_q),\widetilde{C}_\Gam(\an{\n_q},\d_q)\ne
 \emptyset,q\in\N, \lim_{\m_q\to\infty} \frac{1}{\vol(\m_q)}
 \bc_q=\p, \lim_{\n_q\to\infty} \frac{1}{\vol(\n_q)}
 \d_q=\q,\\
 \lim_{q\to\infty} \frac{\log \#
 \widetilde{C}_\Gam(\an{\m_q},\bc_q)}{\vol(\m_q)}\ge h_\Gam(\p)-\varepsilon,
 \;\lim_{q\to\infty} \frac{\log \#
 \widetilde{C}_\Gam(\an{\n_q},\d_q)}{\vol(\n_q)}\ge h_\Gam(\q)-\varepsilon.
 \end{eqnarray*}
 For $i,j\in\N$ we show that

 \begin{equation}\label{ephgcon}
 h_\Gam( \frac{i}{i+j}\p+\frac{j}{i+j}\q)\ge \frac{i}{i+j}h_\Gam(\p)
 +\frac{j}{i+j}h_\Gam(\q)-\varepsilon.
 \end{equation}

 Observe first that for any $\m,\n\in\N^d$ and $\bc\in
 \Pi_n(\vol(\m))$ one has the inequality:
 \begin{equation}\label{frcompin}
 \#\widetilde{C}_\Gam(\an{\n\cdot\m}, \vol(\n)\bc)\ge (\#\tilde
 C_\Gam(\an{\m},\bc))^{\vol(\n)}.
 \end{equation}
 Indeed, view as above, the box
 $\an{\n\cdot\m}$ as a disjoint union of $\vol(\n)$ boxes
 $\an{\m}$.  Color each box $\an{\m}$ in some color in the set
 $\widetilde{C}_\Gam(\an{\m},\bc)$.  Such a coloring is a member of
 $\widetilde{C}_\Gam(\an{\n\cdot\m},\vol(\n)\bc)$.
 Hence (\ref{frcompin}) holds.

 Let $\k_{1,q}:=(m_{1,q}n_{1,q},\ldots,m_{d-1,q}n_{d-1,q}),
 \k_q:=(\k_{1,q},
 (i+j)m_{d,q}n_{d,q})$.  View $\an{\k_q}$ composed of $(i+j)$
 boxes $\an{\m_q\cdot\n_q}$.  The above arguments show that
 \begin{eqnarray*}
 &&\#\widetilde{C}_\Gam(\an{\k_q},\vol(\k_{1,q})(i\bc_q+j\d_q)\ge\\
 &&\#\tilde
 C_\Gam(\an{(\k_{1,q},im_{d,q}n_{d,q}},\vol(\k_{1,q})i\bc_q)
  \#\tilde
 C_\Gam(\an{(\k_{1,q},jm_{d,q}n_{d,q}},\vol(\k_{1,q})j\d_q)\ge\\
 &&(\#\widetilde{C}_\Gam(\an{\m_q},\bc_q))^{i\vol(\n_q)}
 (\#\widetilde{C}_\Gam(\an{\n_q},\d_q))^{j\vol(\m_q)}.
 \end{eqnarray*}
 Since $C_\Gam(\an{\k_q},\vol(\k_{1,q})(i\bc_q+j\d_q)\supset
 \widetilde{C}_\Gam(\an{\k_q},\vol(\k_{1,q})(i\bc_q+j\d_q)$
 by considering the first coloring sequence in this inclusion, and
 using the maximal characterization of $h_\Gam(
 \frac{i}{i+j}\p+\frac{j}{i+j}\q)$,
 we deduce (\ref{ephgcon}).  Since $\varepsilon$ was an arbitrary
 positive number we deduce (\ref{ephgcon}) with $\varepsilon=0$.
 Since $h_\Gam$ is upper semi-continuous we deduce the inequality
 $h_\Gam(a\p+(1-a)\q)\ge ah_\Gam(\p) + (1-a)h_\Gam(\q)$ for any
 $a\in [0,1]$.

 We now prove the claims (\ref{3frnd}-\ref{4frnd}).
 Assume first that Let $\u \in \Der P_\Gam$.  Then $\bpi(\u) = \{\nabla
 P_\Gam(\u)\}=\partial P_\Gam(\u)$ and our assertions trivially
 hold.  Assume next that the assumptions of part (\ref{threebis})
 of  Theorem \ref{maxchar} hold. Recall that $S(\u)
 \subseteq \bpi(\u)$ and $\conv S(\u)=\partial P_\Gam(\u)\supseteq
 \bpi(\u)$.  Let $\p_i\in S(\u), i=1,\ldots,j$.  So
 $P_\Gam(\u)=\p_i\trans \u+ h_\Gam(\p_i),i=1,\ldots,j$.
 Since $\Pi_\Gam$ is convex, we obtain that for any
 $\a=(a_1,\ldots,a_j)\in\Pi_j$ $\p:=\sum_{i=1}^j a_i\p_i\in
 \Pi_\Gam$.  As $h_\Gam$ concave we deduce
 $$P_\Gam(\u)=\sum_{i=1}^j a_i\p_i\trans \u +h_\Gam(\p_i)\le
 \p\trans \u+h_\Gam(\p).$$
 The maximal characterization (\ref{maxchar1}) yields that
 $P_\Gam(\u)=\p\trans \u+h_\Gam(\p)$.  So $\p\in  \bpi(\u)$
 and $h_\Gam(\p)=\sum_{i=1}^j a_ih_\Gam(\p_i)$.
 This proves (\ref{3frnd}-\ref{4frnd}).

 We now prove (\ref{5frnd}).  Recall that $\p\in\partial
 P_\Gam(\R^n)$ if and only if $\p\in \partial
 P_\Gam(\u)$ for some $\u\in\R^n$.  Use part (\ref{3frnd}) of this Theorem and
 part (\ref{six}) of Theorem \ref{maxchar} to deduce the equality
 $h_\Gam(\p)=-P^*_\Gam(\p)$.  If $\Pi_\Gam$ consists of one point
 then $\Pi_\Gam=\partial
 P_\Gam(\R^n)$ and (\ref{5frnd}) trivially holds.
 Assume that $\Pi_\Gam$ consists of more than one point.
 Since $\partial P_\Gam(\R^n)
 \supseteq \ri (\dom P_\Gam^*)\ne \emptyset$, use the second part of
 (\ref{1frnd}) of this Theorem
 to deduce $h_\Gam(\p)=-P^*_\Gam(\p)$
 for each $\p\in \ri (\Pi_\Gam)$.  Suppose that $\q\in\Pi_\Gam\backslash
 \ri (\Pi_\Gam), \p\in  \ri (\Pi_\Gam)$.  Let
 $$f(a):=-h_\Gam(a\q+(1-a)\p),g(a):=P_\Gam^*(a\q+(1-a)\p),
 \textrm{ for } a\in [0,1].$$
 Since $a\q+(1-a)\p\in \ri(\Pi_\Gam)$ for $a\in [0,1)$ it follows
 that $f(a)=g(a)$ for $a\in [0,1)$.
 Since $P_\Gam^*$ is a
 proper closed function, it is lower semi-continuous.  Hence
 $\Pi_\Gam(\q)=g(1)\le \liminf_{a\nearrow 1} g(a)$.  Since $g$ is a convex
 function on $[0,1]$ it follows that $\liminf_{a\nearrow 1} g(a)=
 \lim_{a\nearrow 1} g(a)\le g(1)$.  Hence $g(1)=\lim_{a\nearrow 1} g(a)$.
 Recall that $h_\Gam$ is a concave upper semi-continuous on $\Pi_\Gam$.
 Hence $-h_\Gam$ is a convex lower semi-continuous function on
 $\Pi_\Gam$.  Hence $-h_\Gam(\q)=f(1)=\lim _{a\nearrow 1} f(a)$.
 Therefore $f(1)=g(1)$, i.e. $h_\Gam(\q)=P_\Gam^*(\q)$.
 \end{proof}
  We now list several facts which are consequences of Theorem
  \ref{maxchar}. Given $\u\in\R^n$, then by (\ref{five})
  each $\p \in \bpi(\u)$, namely each $\p$ achieving the maximum in
  (\ref{maxchar1}), is a possible density of the $n$ colors in an
  allowable configuration from $C_\Gam(\Z^d)$ with the potential
  $\u$. That is, the relative frequency of color $i$ is equal to
  $p_i$. For each $\u$ where $P_{\Gam}$ is differentiable, there
  exists a unique density of the $n$ colors. Assume
  that $P_\Gam$ is not differentiable at $\u$. Then $\partial
  P_\Gam$ consists of more than one point. Let $S(\u)$ be defined as
  in (\ref{threebis}). Since $\partial P_\Gam(\u)=\conv S(\u)$,
  $S(\u)$ consists of more than one point. Hence by (\ref{threebis})
  $\bpi(\u)$ consists of more than one point, that is to say, there
  is more than one density for $\u$. In this case $\u$ is called a
  point of \emph{phase transition}, sometimes called a phase
  transition point of the \emph{first order}.

  \begin{prop}\label{heid}  Let $\e := (1,\ldots,1)\trans\in\R^n$.  Then
  for all $t \in \R$
  \[P_\Gam(\u) = t + P_\Gam(\u - t\e).\]
  \end{prop}
  \begin{Proof}  Recall the definition of $Z_\Gam(\m,\u)$ given in (\ref{wmufor}).
  Clearly $\bc(\phi)\trans \e = \vol(\m)$.
  Hence
  \begin{multline}
  Z_\Gam(\m,\u) = \sum_{\phi \in C_\Gam(\an{\m})} e^{\bc(\phi)\trans \u} =
  \sum_{\phi \in C_\Gam(\an{\m})} e^{\bc(\phi)\trans t\e+ \bc(\phi)\trans (\u - t\e)}
  \\
  = e^{t\vol(\m)} Z_\Gam(\m,\u - t\e),
  \end{multline}
  which implies the proposition.  \end{Proof}

  Thus to study $P_\Gam$, we may restrict attention to those
  potentials $\u = (u_1,\ldots,u_n)\trans$ that satisfy $u_n = 0$. (I.e. we reduce
  the number of variables in the function $P_\Gam$ to $n-1$.)
  We show that the same holds for $\partial P_\Gam$ and $\nabla
  P_\Gam$. For $\u \in \R^n$, we use the notation
  \[\overline{\u} := (u_1,\ldots,u_{n-1})\trans\]
  for the projection of $\u$ on the first $n-1$ coordinates, and
  extend it naturally to sets $\overline{\U} := \{\overline{\u} : \u
  \in \U\}$ for $\U \subseteq \R^n$. In the other direction, for $\x
  \in \R^{n-1}$, we use the notation
  \[\iota(\x) := (x_1,\ldots,x_{n-1},1-x_1-\cdots-x_{n-1})\trans\]
  for the unique lifting of $\x$ to the hyperplane $\Sigma_n := \{\x
  \in \R^n : \x\trans \e = 1\}$, and
  \[\iota_0(\x) := (x_1,\ldots,x_{n-1},0).\]
  A straightforward computation shows that
  \begin{equation}\label{straightforward}
   \q\trans(\overline{\z} - z_n \overline{\e}) = \iota(\q)\trans \z -
   z_n \qquad \forall  \q \in \R^{n-1}, \z \in \R^n.
  \end{equation}

  We define the convex function
  $\widehat{P}_\Gam(\cdot)$ on $\R^{n-1}$ by
  \begin{equation}\label{hhat}
  \widehat{P}_\Gam(\x) := P_\Gam(\iota_0(\x)) \qquad \x \in \R^{n-1}.
  \end{equation}
  By taking $t = u_n$ in Proposition~\ref{heid}, we obtain
  \begin{equation}\label{shifting}
   P_\Gam(\u) = u_n + \widehat{P}_\Gam(\overline{\u} - u_n \overline{\e}).
  \end{equation}
  We now obtain
  a straightforward generalization of the density theorem proved
  in~\cite{FP} for monomer-dimer tilings.
  %
  \begin{theo}\label{denthm}
  Let $\widehat{P}_\Gam(\cdot)$ be defined on $\R^{n-1}$
  by~(\ref{hhat}).  Then
  \begin{equation}\label{denthmeq1}
  \overline{\partial P_\Gam(\u)} = \partial \widehat{P}_\Gam(\overline{\u} -
  u_n\overline{\e})) \qquad \forall \u \in \mathbb{R}^n,
  \end{equation}
  \begin{equation}\label{denthmeq2}
  \overline{\partial P_\Gam(\R^n)} = \partial \widehat{P}_\Gam(\R^{n-1}).
  \end{equation}
  Furthermore, $P_\Gam$ is differentiable at $\u$ if and only if
  $\widehat{P}_\Gam$ is differentiable at $\overline{\u} -
  u_n\overline{\e}$. If $\widehat{P}_\Gam$ has all $n-1$ partial
  derivatives at $\overline{\u} - u_n \overline{\e}$, then $P_\Gam$ is
  differentiable at $\u$ and
  \begin{equation}\label{uniqueLambdahat}
  \nabla P_\Gam(\u) = \left(\frac{\partial \widehat{P}_{\Gam}}{\partial u_1}
  (\overline{\u} - u_n\overline{\e}),\ldots,
  \frac{\partial \widehat{P}_{\Gam}}{\partial u_{n-1}}(\overline{\u} - u_n\overline{\e}),
  1 - \sum_{i=1}^{n-1} \frac{\partial \widehat{P}_{\Gam}}{\partial u_i}
  (\overline{\u} -  u_n\overline{\e})\right).
  \end{equation}
  \end{theo}
  \begin{Proof}
  Assume that $\p \in \partial P_\Gam(\u)$. By definition of $\partial
  P_\Gam(\u)$ we have
  \begin{equation}\label{support}
   P_\Gam(\u + \v) \geq \p\trans \v + P_\Gam(\u) \qquad \forall \v.
  \end{equation}
  Choose $\v$ such that $v_n = 0$. Then,
  using (\ref{shifting}) for $P_\Gam(\u + \v)$ and $P_\Gam(\u)$
  in~(\ref{support}), we obtain
  \[u_n + \widehat{P}_\Gam(\overline{\u} - u_n \overline{\e} + \overline{\v}) \geq
  \overline{\p}\trans \overline{\v} + u_n + \widehat{P}_\Gam(\overline{\u} - u_n
  \overline{\e}) \qquad \forall \overline{\v},
  \]
  which by definition means $\overline{\p} \in \partial \widehat{P}_\Gam(\overline{\u} -
  u_n\overline{\e})$.

  Conversely, assume that $\q \in \partial \widehat{P}_\Gam(\overline{\u} -
  u_n\overline{\e})$, which means that
  \begin{equation}\label{supportdown}
  \widehat{P}_\Gam(\overline{\u} - u_n \overline{\e} + \overline{\v}) \geq
  \q\trans \overline{\v} + \widehat{P}_\Gam(\overline{\u} - u_n
  \overline{\e}) \qquad \forall \overline{\v}.
  \end{equation}
  Now for arbitrary $\z$, choose $\overline{\v} = \overline{\z} - z_n
  \overline{\e}$ in~(\ref{supportdown}), and use~(\ref{straightforward}) once and
  (\ref{shifting}) twice in the resulting inequality to obtain
  \[P_\Gam(\u + \z) - u_n - z_n \geq \iota(\q)\trans\z - z_n + P_\Gam(\u) -
  u_n \qquad \forall \z,
  \]
  which means by definition that $\iota(\q) \in \partial P_\Gam(\u)$, and
  therefore $\q \in \overline{\partial P_\Gam(\u)}$. We have
  proved~(\ref{denthmeq1}).

  We now show (\ref{denthmeq2}). It follows easily from
  (\ref{denthmeq1}) that $\overline{\partial P_\Gam(\R^n)} \subseteq
  \partial \widehat{P}_\Gam(\R^{n-1})$. To show the reverse
  inclusion, let $\q \in \partial \widehat{P}_\Gam(\w)$, and let $\u
  = \iota_0(\w)$ so that $\w = \overline{\u} - u_n \overline{\e}$.
  By (\ref{denthmeq1}) we have $\q \in \overline{\partial
  P_\Gam(\u)}$, as required.

  Assume that $P_\Gam$ is differentiable at $\u$. Then $\partial
  P_\Gam(\u)$ is a singleton, and therefore $\overline{\partial
  P_\Gam(\u)}$ is a singleton. By (\ref{denthmeq1}) $\partial
  \widehat{P}_\Gam(\overline{\u} - u_n\overline{\e}))$ is a
  singleton, and therefore $\widehat{P}_\Gam$ is differentiable at
  $\overline{\u} - u_n\overline{\e}$.

  Conversely assume that $\widehat{P}_\Gam$ is differentiable at
  $\overline{\u} - u_n\overline{\e}$, and therefore $\partial
  \widehat{P}_\Gam(\overline{\u} - u_n\overline{\e})$ is a singleton
  $\{\q\}$. We assert that $\partial P_\Gam(\u) = \{\iota(\q)\}$.
  Indeed, if $\x \in \partial P_\Gam(\u)$, then $\overline{\x} \in
  \overline{\partial P_\Gam(\u)} = \partial
  \widehat{P}_\Gam(\overline{\u} - u_n\overline{\e})  = \{\q\}$,
  where the first equality is  by (\ref{denthmeq1}), and so
  $\overline{\x} = \q$. Since $\x \in \partial P_\Gam(\u) \subseteq
  \Pi_n \subset \Sigma_n$ by~(\ref{one}) of Theorem~\ref{maxchar},
  it follows that $\x = \iota(\q)$, proving the assertion since
  $\partial P_\Gam(\u) \neq \emptyset$. Therefore $P_\Gam$ is
  differentiable at $\u$.

  The last statement of the theorem follows from the previous
  statement and the fact \cite[Thm 25.2]{Roc} that a convex function
  $f : \R^k \to \R$ is differentiable at a point $\a$ if merely the
  $k$ partial derivatives $\frac{\partial f}{\partial x_i}$, $i
  =1,\ldots,k$ exist at $\a$. \end{Proof}
  We can reformulate Theorem \ref{maxchar} for
  $\widehat{P}_\Gam$.

   \begin{theo}\label{maxcharhat}
   Let $\widehat{P}_\Gam^*$ be the conjugate convex function of
   $\widehat{P}_\Gam$.
   Then
   { 
   \renewcommand{\theenumi}{\alph{enumi}}
   \renewcommand{\labelenumi}{(\theenumi)}
   \begin{enumerate}
     \item \label{twoiota} $h_\Gam(\p) \leq -\widehat{P}^*_\Gam(\overline{\p})$
     for all $\p \in \Pi_\Gam$.
     \item \label{fouriota}
      \begin{equation}\label{maxchar1hat}
      \widehat{P}_\Gam(\x) = \max_{\p \in \Pi_{\Gam}}
      (\overline{\p}\trans\x + h_{\Gam}(\p)) \text{ for all } \x\ \in
      \R^{n-1}.
      \end{equation}
      For $\x \in \R^{n-1}$, let
      $\q(\x) \in \overline{\Pi_\Gam(\iota_0(\x))}$
      i.e., $\q(\x)$
      is any vector satisfying
      \begin{equation}\label{defpx}
      \begin{split}
      \iota (\q(\x))\in \Pi_\Gam \;\;\text{ and }\;\; \widehat{P}_\Gam(\x) & =
      \iota(\q(\x))\trans \iota_0(\x) + h_{\Gam}(\iota (\q(\x)))
      \\
      & = {\q(\x)}\trans \x + h_{\Gam}(\iota (\q(\x))).
      \end{split}
      \end{equation}
     \item \label{sixiota} $h_\Gam(\iota (\q(\x)))= -\widehat{P}^*_\Gam(\q(\x))$.
     \item \label{fiveiota} $\q(\x)\in \partial \widehat{P}_{\Gam}(\x)$. In
     particular, if $\x \in \Der \widehat{P}_\Gam$, then $\q(\x) = \nabla
     \widehat{P}_\Gam(\x)$. Therefore $\partial \widehat{P}_\Gam
     (\Der \widehat{P}_\Gam) \subseteq
     \overline{\Pi_\Gam}$.
     \item \label{threebisiota}
     Let $\x \in \R^{n-1} \setminus \Der \widehat{P}_\Gam$, and let $S(\x)$
     consist of all the limits of sequences $\nabla \widehat{P}_\Gam(\x_i)$
     such that $\x_i \in \Der \widehat{P}_\Gam$ and $\x_i \to \x$. Then $S(\x)
     \subseteq \overline{\Pi_\Gam(\iota_0(\x))}$.
     \item \label{oneiota} $\conv \overline{\Pi_\Gam} = \dom \widehat{P}_\Gam ^*$.
   \end{enumerate}
   } 
  \end{theo}
  \begin{Proof}
  From (\ref{maxchar1}) we have $P_\Gam(\u) \geq \p\trans \u +
  h_{\Gam}(\p)$ for all $\u \in \R^n$ and $\p \in \Pi_{\Gam}$. Fix
  $\p\in\Pi_\Gam$, and let $\x \in \R^{n-1}$ and $\u =
  \iota_0(\x)$. Then $\widehat{P}_\Gam(\x)=P_\Gam(\u) \geq
  \overline{\p}\trans\x +h_\Gam(\p)$, so $-h_\Gam (\p) \geq
  \overline{\p}\x - \widehat{P}_\Gam(\x)$.  Now take the supremum
  over $\x \in\R^{n-1}$ to obtain $-h_\Gam(\p) \geq
  \widehat{P}_\Gam^*(\overline{\p})$, which is (a).  Substitute
  $\u=\iota_0(\x)$ in (\ref{maxchar1}) to deduce
  (\ref{maxchar1hat}). We now show (\ref{sixiota}). By
  (\ref{defpx}) and the definition of $\widehat{P}^*$ it follows
  that $-h_\Gam(\iota(\q(\x))) = \q(\x)\trans\x -
  \widehat{P}_\Gam(\x) \leq \widehat{P}_\Gam^*(\q(\x))$. Combining
  this with the opposite inequality (\ref{twoiota}), we deduce
  (\ref{sixiota}). To prove (\ref{fiveiota}), let $\x,\z \in
  \R^{n-1}$. Since $\iota(\q(\x)) \in \Pi_\Gam$,
  (\ref{maxchar1hat}) applied to $\x + \z$ and (\ref{defpx}) give
  $\widehat{P}_\Gam(\x+\z) \geq \q(\x)\trans(\x+\z) +
  h_\Gam(\iota(\q(\x))) = \q(\x)\trans \z + \widehat{P}_\Gam(\x)$.
  This proves (\ref{fiveiota}). To prove (\ref{threebisiota}), let
  $\x \in \R^{n-1} \setminus \Der \widehat{P}_\Gam$ and let $\q
  \in S(\x)$. Then there exists a sequence $\x_i \in \Der
  \widehat{P}_\Gam$ such that $\x_i \to \x$ and $\nabla
  \widehat{P}_\Gam(\x_i) \to \q$. By the ``furthermore'' part of
  Theorem~\ref{denthm} we have $\iota_0(\x_i) \in \Der P_\Gam$ and
  $\iota_0(\x )\in \R^n \setminus \Der P_\Gam$.
  By~(\ref{uniqueLambdahat}) and the continuity of $\iota$ we have
  $\nabla P_\Gam(\iota_0(\x_i)) = \iota(\nabla
  \widehat{P}_\Gam(\x_i)) \to \iota(\q)$. This shows that
  $\iota(\q) \in S(\iota_0(\x)) \subseteq \Pi_\Gam(\iota_0(\x))$,
  where the inclusion is by (\ref{threebis}) of
  Theorem~\ref{maxchar}. It follows that $\q =
  \overline{\iota(\q)} \in \overline{\Pi_\Gam(\iota_0(\x))}$, as
  required. Now we prove (\ref{oneiota}). By (\ref{twoiota}) we
  have $\widehat{P}^*_\Gam(\overline{\p}) \leq -h_\Gam(\p) <
  \infty$ for all $\p \in \Pi_\Gam$, so $\overline{\Pi_\Gam}
  \subseteq \dom \widehat{P}_\Gam ^*$, and taking the convex hull
  gives $\conv \overline{\Pi_\Gam} \subseteq \dom \widehat{P}_\Gam
  ^*$. It remains to show the reverse inclusion. We know by
  Lemma~\ref{ransubf} that $\partial \widehat{P}_\Gam(\R^{n-1})$
  is the set of all points where $\widehat{P}^*_\Gam$ is
  subdifferentiable, so in particular $\ri(\dom
  \widehat{P}^*_\Gam) \subseteq \partial
  \widehat{P}_\Gam(\R^{n-1})$. By (\ref{denthmeq2}) and
  (\ref{one}) of Theorem~\ref{maxchar} we have $\partial
  \widehat{P}_\Gam(\R^{n-1}) = \overline{\partial P_\Gam(\R^n)}
  \subseteq \overline{\conv \Pi_\Gam} = \conv
  \overline{\Pi_\Gam}$. Combining the above inclusions gives
  $\ri(\dom \widehat{P}^*_\Gam) \subseteq \conv
  \overline{\Pi_\Gam}$. Applying the closure operator, we obtain
  $\dom \widehat{P}^*_\Gam \subseteq \cl(\dom \widehat{P}^*_\Gam)
  = \cl(\ri(\dom \widehat{P}^*_\Gam)) \subseteq \conv
  \overline{\Pi_\Gam}$ as in  the proof of (\ref{onebis}) of
  Theorem~\ref{maxchar}.
  \end{Proof}

  Since a probability vector $\p\in\Pi_n$ is determined completely by its projection $\bar \p$ on the first $n-1$
  components, we can view the function $h_\Gam: \Pi_\Gam \to \R_+$ as a function on $\overline{\Pi_\Gam}$.
  Formally, let
  \begin{equation}\label{defhbar}
  \bar h_\Gam(\q): = h_\Gam(\iota(\q))  \textrm{ for all } \q\in \overline{\Pi_\Gam}.
  \end{equation}
 \section{$P_\Gam$ and density entropies for one dimensional
 SOFT}\label{sec:onedimensionalsoft}

 In this section we apply the results of Section~\ref{densityentropy}
 to one dimensional SOFT.  In this case $\Gam$ is given
 by a digraph $\Gamma:=\Gamma_1 \subseteq \an{n}\times \an{n}$.

 \begin{theo}\label{present1soft}
 Let $\Gamma\subseteq \an{n}\times\an{n}$ be a digraph
 on $n$ vertices, with at least one strongly connected component.
 Then $P_{\Gamma}(\u)=\log \rho(D(\Gamma,\u))$,
 where the nonnegative matrix $D(\Gamma,\u)$ is given in
 Proposition \ref{onedimpressure}.  If $\Gamma$ is strongly
 connected, or more generally $\Gamma$ has one connected
 component, then $P_{\Gamma}$ is an analytic function on $\R^n$,
 $\Pi_{\Gamma}$ is a closed convex set of probability vectors equal
 to $\dom P_{\Gamma}^*$,
 $h_{\Gamma}$ is concave and continuous on $\ri
 \Pi_{\Gamma}$, and coincides there with $-P_{\Gamma}^*$.
 In particular, for any $\u\in\R^n$
 \begin{equation}\label{hgamfor}
 h_{\Gamma}(\p(\u))= -\p(\u) \trans \u
 +\log\rho(D(\Gamma,\u)),\;  \textrm{where } \p(\u):=\frac{\nabla \rho(D(\Gamma,\u))}
 {\rho(D(\Gamma,\u))}.
 \end{equation}
 Furthermore
 \begin{equation}\label{hgamfor1}
 h_{\Gamma}=\max_{\p\in\Pi_{\Gamma}} h_{\Gamma}(\p)=
 h_{\Gamma}(\frac{\nabla\rho(D(\Gamma,\0))}{\rho(\Gamma)})=\log\rho(\Gamma).
 \end{equation}

 Assume that $\Gamma$ has $k>1$ connected components
 $\Delta_1,\ldots,\Delta_k$.  Then $P_{\Gamma}(\u)=$

 \noindent
 $\max (P_{\Delta_1}(\u),\ldots,P_{\Delta_k}(\u))$, where
 each $P_{\Delta_i}$ is an analytic function on $\R^n$.

 \end{theo}
 \begin{Proof}  Proposition \ref{onedimpressure} yields the
 equality $P_{\Gamma}(\u)=\log \rho(D(\Gamma,\u))$.
 Assume first that $\Gamma$ is strongly connected,
 which is equivalent to the assumption that the
 adjacency matrix $D(\Gamma)=(d_{ij})_{i,j\in \an{n}}$ is irreducible.
 From the definition of $D(\Gamma,\u) =
 (d_{ij} \cdot  e^{\frac{1}{2}(\e_i\trans \u + \e_j\trans \u)})_{i,j\in\an{n}}$
 it follows that $D(\Gamma,\u)$ is an irreducible matrix
 for each value $\u\in\R^n$.  Then $\rho(D(\Gamma,\u))$
 is a simple root of the characteristic equation
 $p(z,\u):=\det (zI-D(\Gamma,\u))=0$ for each $\u\in\R^n$.
 Since the coefficients of $p(z,\u)$ are analytic in $\u$,
 where $\u\in\C^n$, the implicit function theorem implies
 that $\rho(D(\Gamma,\u))$ is analytic function in a neighborhood
 of $\R^n$ of $\C^n$.  Since $\rho(D(\Gamma,\u))$ is positive
 on $\R^n$ it follows that $\log\rho(D(\Gamma,\u)), \u\in\R^n$
 has an analytic extension to some neighborhood of $\R^n$
 in $\C^n$.  Hence $P_{\Gamma}$ is analytic on $\R^n$.
 In particular, $P_{\Gamma}$ is differentiable on $\R^n$.
 Theorem \ref{maxchar} and Lemma \ref{ransubf} yield that
 $\dom P_{\Gamma}^*=\conv \Pi_{\Gamma}\supseteq\Pi_{\Gamma} \supseteq \partial
 P_{\Gamma}(\R^n)\supseteq \ri(\dom P_{\Gamma}^*)$.
 Since $\Pi_{\Gamma}$ is closed we obtain that
 $\Pi_{\Gamma}=\Cl(\dom P_{\Gamma}^*)=\conv \Pi_{\Gamma}$,
 hence $\Pi_{\Gamma}$ is convex.  According to Theorem
 \ref{maxchar} $\Cl(\dom P^*_{\Gam})=\dom P^*_{\Gam}$.

 Since $P_{\Gamma}$ is differentiable, Theorem  \ref{maxchar} yields that
 $h_{\Gamma}(\p)=-P^*_{\Gamma}(\p)$ for $\p\in \partial
 P_{\Gamma}(\R^n)$.  As $\partial P_{\Gamma}(\R^n) \supseteq \ri
 (\dom P_{\Gamma}^*)$ we deduce that $h_{\Gamma}=-P^*_{\Gamma}$ on
 $\ri (\dom P_{\Gamma}^*)$.  Since $P^*_{\Gamma}$ is a convex
 continuous function on $\ri (\dom P_{\Gamma}^*)$, it follows that
 $h_{\Gamma}$ is a concave continuous function on $\ri (\dom
 P_{\Gamma}^*)$.

 As $P_{\Gamma}(\u)=\log\rho(D(\Gamma,\u))$ it follows that
 $\nabla P_{\Gamma}(\u)=\frac{\nabla \rho(D(\Gamma,\u))}{\rho(D(\Gamma,\u))}$.
 Hence (\ref{hgamfor}) holds.  Clearly,
 $\rho(D(\Gamma,\0))=\rho(\Gamma)$ and (\ref{hgamfor1}) follows.

 Assume now that $\Gamma$ is not strongly connected digraph.
 Rename the vertices of $\Gamma$ such that $D(\Gamma)$ is its
 \emph{normal} form \cite[XIII.4]{Gan}.   That is, $D(\Gamma)$ is a block lower
 triangular form matrix, where
 each submatrix on a diagonal block is either a nonzero irreducible matrix or
 $1\times 1$ zero matrix.  Then each nonzero irreducible
 submatrix corresponds to a strongly irreducible component of
 $\Gamma$. Let $\Delta_1,\ldots,\Delta_k$ be the $k\ge 1$
 irreducible components of $\Gamma$.  Since $D(\Gamma,\u)$
 is also in its normal form it follows that
 $\rho(D(\Gamma,\u))=\max_{i\in [1,k]} \rho(D(\Delta_k,\u))$.
 Note that $\log\rho(D(\Delta_i,\u))=P_{\Delta_i}(\u)$ for
 $i=1,\ldots,k$.

 Assume first that $k=1$ and $\Delta_1\ne \Gamma$.  Rename
 the vertices of $\Gamma$ such that $\an{m}$ is the set of vertices
 of $\Delta_1$, where $1\le m <n$.  Let $\tilde \u=(u_1,\dots,
 u_m)\trans$.  Then $P_{\Gamma}(\u)=P_{\Delta_1}(\tilde \u)$
 and the theorem follows in this case.

 Assume finally that $k>1$.  The above arguments show that each $P_{\Delta_i}$
 is an analytic function in $\u$, which does not depend on a
 variable $u_j$ if $j$ is not a vertex of $\Delta_i$.
 \end{Proof}

 Assume that $\Gamma$ is strongly connected and we want to compute
 $\p(\u)=\frac{\nabla \rho(D(\Gamma,\u))} {\rho(D(\Gamma,\u))}$.
 We give the following simple formula for $\p(\u)$ which is known
 to the experts.

 \begin{prop}\label{radru}
 Let $\Gamma\subseteq \an{n}\times\an{n}$ be
 a strongly connected digraph
 on $n$ vertices.  Let $D(\Gamma,\u)$ be the nonnegative matrix given in
 Proposition \ref{onedimpressure}.  Let $\x(\u)=(x_1(\u),\ldots,x_n(\u))\trans$,
 $\y(\u)=(y_1(\u),\ldots,y_n(\u))\trans$
 be positive eigenvectors of $D(\Gamma,\u)$, $D(\Gamma,\u)\trans$
 respectively, normalized by the condition $\y(\u)\trans \x(\u)=1$.
 Then
 \begin{equation}\label{radru1}
 \nabla P_{\Gamma}(\u)=\frac{\nabla \rho(D(\Gamma,\u))} {\rho(D(\Gamma,\u))}
 =(y_1(\u)x_1(\u),\ldots,y_n(\u)x_n(\u)) \textrm{ for each }
 \u\in\R^n.
 \end{equation}
 \end{prop}
 \begin{Proof}  Let $D(\u):=D(\Gamma,\u),\rho(\u):=\rho(D(\Gamma,\u))$.
 Since $\rho(\u)>0$ is a simple root of $\det (zI -
 D(\u))$ it follows that one can choose $\x(\u),\y(\u)$ to be
 analytic on $\R^n$ in $\u$.  (For example first choose
 $\x(\u),\tilde\y(\u)\in \R^n_+$ to be the unique left and right eigenvectors of $D(\u)$
 of length $1$.  Then let $\y(\u)=\frac{1}{\tilde\y(\u)\trans \x(\u)}
 \y(\u)$.)  Let $\partial_i$ be the partial derivative with respect
 to $u_i$.  Then
 $$\y(\u)\trans \x(\u)=1 \forall \u\in \R^n
 \Rightarrow\partial_i\y(\u)\trans \x(\u) +\y(\u)\trans
 \partial_i\x(\u)=0, \texttt{ for }i=1,\ldots,n.$$
 Observe next that $\y(\u)\trans D(\u)\x(\u)=\rho(\u)$.
 Taking the partial derivative with respect to $u_i$
 and using the formula (\ref {AGamma1u}) for the entries of $D(\u)$ we
 obtain
 \begin{eqnarray*}
 &&\partial_i\rho(\u)= \partial_i\y(\u)\trans D(\u) \x(\u)
 +\y(\u)\trans D(\u)
 \partial_i\x(\u) +\y(\u)\trans \partial_i D(\u) \x(\u).\\
 &&\rho(\u)(\partial_i\y(\u)\trans \x(\u) +\y(\u)\trans
 \partial_i\x(\u)) + \rho(\u)y_i(\u)x_i(\u)= \rho(\u)y_i(\u)x_i(\u).
 \end{eqnarray*}
 This proves (\ref{radru1}).  \end{Proof}

 We now apply the above results to the following simple digraph
 on two vertices:

 \begin{center}
  \psset{unit=1cm}
  \begin{pspicture}(2,.5)
   \psset{fillstyle=solid,radius=.1}
   \Cnode[fillcolor=red](0,0){red}
   \Cnode[fillcolor=blue](1,0){blue}
   \ncarc[arcangle=30]{->}{red}{blue}
   \ncarc[arcangle=30]{->}{blue}{red}
   \nccircle[angle=-90]{->}{blue}{.2}
   \Cnode[fillcolor=blue](1,0){anotherblue}
  \end{pspicture}
 \end{center}
 Identify the red color with the state $1$ and the blue color with
 the state $2$, which is usually identified with the state $0$.
 Then $C_{\Gamma}(\Z)$ consists of all coloring of the lattice
 $\Z$ in blue and red colors such that no two red colors are
 adjacent.  This is the simplest \emph{hard core model} in
 statistical mechanics.  The adjacency matrix $D(\Gamma)$ is
 the following $2\times 2$ matrix $\left(\begin{array}{cc} 0
 &1\\1&1\end{array}\right)$.
 Let $\u=(s,t)\trans$.  Then $\nabla P_{\Gamma}(\u)=(p_1(\u),p_2(\u))\in \Pi_2$
 it follows that $p_2(\u)=1-p_1(\u)$.  It is enough to consider
 $\u=(s,0)$ and $p_1(s)=\frac{dP_{\Gamma}((s,0)\trans)}{ds}$.
 So $p:=p_1(s)$ is the density of $1$ in all the configurations
 of infinite strings of $0,1$, where no two $1$ are adjacent.
 Clearly $D(\Gamma,\u)=\left(\begin{array}{cc} 0
 &e^{\frac{s}{2}}\\ e^{\frac{s}{2}}&1\end{array}\right)$.
 Hence
 \begin{eqnarray*}
 &&\rho(\u)=\frac{1+\sqrt{1+4e^s}}{2},\quad p_1(s)=\frac{2e^s}
 {(1+\sqrt{1+4e^s})\sqrt{1+4e^s}}=\\
 &&\frac{2}{(e^{-\frac{s}{2}}+\sqrt{e^{-s}+4})\sqrt{e^{-s}+4}}
 = \frac{1}{2}\big(1-\frac{1}{\sqrt{1+4e^s}}\big)\in (0,\frac{1}{2}).
 \end{eqnarray*}
 Note that $p_1(s)$ is increasing on $\R$, and $p_1(-\infty)=0,
 p_1(\infty)=\frac{1}{2}$.  Hence $\Pi_{\Gamma}=$

 \noindent
 $\conv(\{(0,1)\trans,\frac{1}{2}(1,1)\trans\})$ and $\partial P_{\Gamma}(\R^2)=\ri
 \Pi_{\Gamma}$.
 As $P_{\Gamma}(\0)=h_{\Gamma}=\log \frac{1+\sqrt{5}}{2}$ it follows that the
 value $p^*:=p_1(0)=\frac{2}{(1+\sqrt{5})\sqrt{5}}=.2763932024$
 is the density $p^*$ of $1$'s for
 which $h_{\Gamma}=h_{\Gamma}((p^*,1-p^*))$.

 To find the formula for $\bar h_{\Gamma}(p)=h_{\Gamma}((p,1-p))$ first
 note that if $p=p_1(s)$ then
 $$\sqrt{1+4e^s}=\frac{1}{1-2p}, \quad s(p)=
 \log\frac{p(1-p)}{(1-2p)^2}.$$
 Then
 $$\bar h_{\Gamma}(p)= \log \frac{1-p}{1-2p} -
 p \log\frac{p(1-p)}{(1-2p)^2}, \quad p\in (0,\frac{1}{2}).$$

 Our computations of $P_\Gam$, for $d\ge 2$, are based on
 upper and lower bounds, for example as given in Corollary \ref{ulbd=2}.
 We claim that the function $\frac{\log\theta_2(m,\u)}{m}$ can be viewed
 as the pressure function
 of certain corresponding one dimensional subshift of finite type
 given.

 Consider for the simplicity of the exposition two
 dimensional SOFT given by $\Gam=(\Gamma_1,\Gamma_2)$, where
 $\Gamma_1$ is a symmetric digraph.  Let $\Delta$ be the transfer digraph induced
 by $\Gamma_2$ between the
 allowable $\Gamma_1$ coloring of the circle $T(m)$.
 Then $V:=C_{\Gamma_1,\perio}(m)$ are
 the set of vertices of $\Delta$.  For any
 $\alpha,\beta \in C_{\Gamma_1,\perio}(m)$ the directed
 edge
 $(\alpha,\beta)$ is in $\Delta$ if and only if the configuration
 $[(\alpha,\beta)]$ is an allowable configuration on
 $C_\Gam((m,2))$.  Note that the adjacency matrix $D(\Delta)=(d_{\alpha\beta})
 _{\alpha,\beta\in C_{\Gamma_1,\perio}(m)}$ is $N\times N$
 matrix, where
 $N:=\#C_{\Gamma_1,\perio}(m)$.
 Then the one dimensional SOFT is $C_{\Gam}(T(m)\times \Z)$:
 all $\Gam$ allowable coloring of
 the infinite torus in the direction $\e_2$ with the basis $T(m)$.
 The pressure corresponding to this one dimensional SOFT
 is denoted by $\tilde P_{\Delta}(\u)$.  It is given by the
 following formula:
 Let
 \begin{equation}\label{tildedudef}
 \tilde D(\Delta,\u)=(\tilde d_{\alpha\beta}(\u))
 _{\alpha,\beta\in C_{\Gamma_1,\perio}(m)},\;
 \tilde d_{\alpha\beta}(\u)=d_{\alpha\beta}e^{\frac{1}{2}(\bc(\alpha)+\bc(\beta))\trans
 \u}.
 \end{equation}
 Then
 \begin{equation}\label{tildedelpres}
 \tilde P_{\Delta}(\u):=\frac{\log\rho(\tilde D(\Delta,\u))}{m}.
 \end{equation}
 The reason we divide $\log\rho(\tilde D(\Delta,\u))$ by $m$, is
 to have the normalization
 \[\tilde P_{\Delta}(\u + t\e)=\tilde P_{\Delta}(\u)+t\quad \textrm{for any
 }t\in\R.\]

 It is straightforward to show that $\frac{\log\theta_2(m,\u)}{m}=\tilde
 P_{\Delta}(\u)$.
 Assume that $\Delta$ has one irreducible component.
 Then the arguments of the proof of Proposition \ref{radru}
 yield that $\tilde P_{\Delta}(\u)$ is analytic on $\R^n$.
 Furthermore
 \begin{equation}\label{nabfortilpres}
 \nabla \tilde P_{\Delta}(\u)=(\y(\u)\trans (\partial_1 \tilde
 D(\Delta,\u))\x(\u),\ldots, \y(\u)\trans (\partial_n \tilde
 D(\Delta,\u))\x(\u)),
 \end{equation}
 for any $\u\in\R^n$.  Here
 $\x(\u)$ and $\y(\u)$ are the nonnegative eigenvectors of
 $D(\Delta,\u)$ and $D(\Delta,\u)\trans$, respectively, normalized by
 the condition $\y(\u)\trans \x(\u)=1$.  Then $\nabla \tilde
 P_{\Delta}(\u)\in\Pi_n$ corresponds to the limiting densities
 of the $n$ kind of particles in this one dimensional
 SOFT.

 In the numerical computations, as in the next section, we use
 one dimensional subshifts to estimate the pressure $P_\Gam$
 from above or below as described for example in Corollary \ref{ulbd=2}.  To estimate
 the partial derivatives of $P_\Gam$ one can find the partial
 derivatives of the pressure corresponding to the one
 dimensional subshift approximation using Proposition
 \ref{radru}.  Since $P_\Gam(\u)$ is convex in each variable we
 can estimate each partial derivative from above and below by
 finite differences.  However, these estimates are not as good
 as taking the derivatives of the one dimensional subshift
 approximation to $P_\Gam(\u)$.

 \section{The monomer-dimer model in $\Z^d$}\label{sec:weightedmonomerdimertiligs}

 A \emph{dimer} is a union of two adjacent sites in the grid
 $\Z^d$, and a \emph{monomer} is a single site. By a \emph{tiling}
 of a set $S \subseteq \Z^d$ we mean a partition of $S$ into
 monomers and dimers. By a \emph{cover} of $S$ we mean a tiling of
 a superset of $S$ with each monomer contained in $S$ and each
 dimer meeting $S$; in other words, dimers are allowed to protrude
 halfway out of $S$. Usually our set $S$ will be a box or the
 entire $\Z^d$; in the case of a torus we only speak of tilings.
 As mentioned in \cite{FP}, the set of monomer-dimer
 tilings of $\Z^d$ can be encoded as an NNSOFT $C_\Gam(\Z^d)$ as
 follows. We color $\Z^d$ with the $2d + 1$ colors $1,\ldots,2d +
 1$: a dimer in the direction of $\e_k$ occupying the adjacent
 sites $\i,\i + \e_k$ is encoded by the color $k$ at $\i$ and the
 color $k + d$ at $\i + \e_k$; a monomer at $\i$ is encoded by the
 color $2d + 1$ at $\i$.  This imposes restrictions on the
 coloring, which are expressed by the $d$-digraph $\Gam =
 (\Gamma_1,\ldots,\Gamma_d)$ on the set of vertices $\an{2d + 1}$,
 where
 \begin{equation}\label{encoding}
 (p,q) \in \Gamma_k \Longleftrightarrow (p = k, q = k+d) \text{ or } (p \neq k,q \neq k+d).
 \end{equation}
 It is easy to check that this gives a bijection between the
 monomer-dimer tilings of $\Z^d$ and $C_\Gam(\Z^d)$.
 Let $P_\Gam(\u), \u\in\R^{2d+1}$ be the pressure function for the
 monomer-dimer model in $\Z^d$.  Since each dimer in the direction
 $\e_k$ corresponds to the colors $k$ and $k+d$ it follows that
 $P_\Gam(\u)$ is effectively a function of $d+1$ variables.
 To show that
  we define the following linear transformations
  \begin{defn}\label{defnvart}
  { 
   \renewcommand{\theenumi}{\alph{enumi}}
   \renewcommand{\labelenumi}{(\theenumi)}
  \begin{enumerate}
  \item  Let
  $T,T_1:\R^{d+1}\to \R^{2d+1}$ be the linear transformations
  \begin{eqnarray*}
  T(w_1,\ldots,w_d,w_{d+1})
  =(\frac{w_1}{2},\ldots,\frac{w_d}{2},\frac{w_1}{2},\ldots,\frac{w_d}{2},
  w_{d+1}),\\
  T_1(w_1,\ldots,w_d,w_{d+1})=(w_1,\ldots,w_d,w_1,\ldots,w_d,w_{d+1}).
  \end{eqnarray*}

  \item  Let $Q:\R^{2d+1} \to \R^{d+1}$
  be the linear transformation given by
  $Q(u_1,\ldots,u_{2d+1}) =
  (u_1+u_{d+1},\ldots,u_{d}+u_{2d},u_{2d+1})$.

  \item Let $Q_d:\R^d\to \R$ be the linear transformation $(v_1,\ldots,v_d)\trans
  \mapsto v_1+\dots+v_d$.

  \end{enumerate}

   } 

  \end{defn}

  \begin{theo}\label{mdimertheo}  Let
  $\Gam=(\Gamma_1,\ldots,\Gamma_d)$-coloring, with $2d+1$ colors
  given by (\ref{encoding}) For $\u\in\R^{2d+1}$
  let $P_{\Gam}(\u)$ denote the pressure function.
  Then
  { 
   \renewcommand{\theenumi}{\alph{enumi}}
   \renewcommand{\labelenumi}{(\theenumi)}
  \begin{enumerate}
  \item \label{average}
  $P_\Gam(\u) =
  P_\Gam(TQ\u)$.
  \item\label{pgamprp}

  $\partial P_{\Gam}(\R^{2d+1}) \subseteq T(\R^{d+1})$.

  \item\label{descpgam}  $\Pi_\Gam= T(\Pi_{d+1})$.
  Hence $\partial P_{\Gam}(\R^{2d+1}) \subseteq T(\Pi_{d+1})$.

  \item\label{conchgam} The function $h_\Gam: T\Pi_{d+1}\to \R_+$
  is a concave function.


  \end{enumerate}
 } 

  \end{theo}

 \begin{Proof}
 { 
 \renewcommand{\theenumi}{\alph{enumi}}
 \renewcommand{\labelenumi}{(\theenumi)}
  \begin{enumerate}
  \item  Since the colors $i$ and $i+d$ describe
  the two halves of a dimer in the direction $\e_i$ for
  $i=1,\ldots,d$, we have the identity $P_\Gam(\u) =
  P_\Gam(TQ\u)$.

  \item Let $\p=(p_1,\ldots,p_{2d+1}) \in \partial(P_\Gam(\u))$. In case $\u \in \Der
  P_\Gam$, the Chain Rule applied to the identity in~(\ref{average}) yields the equalities
  $p_i=p_{d+i}, i=1,\ldots,d$.
  In case $\u \in \partial(P_\Gam(\u)) \setminus \Der
  P_\Gam$, this follows from the fact that $\partial P_\Gam(\u) =
  \conv S(\u)$ as in the beginning of
  Section~\ref{densityentropy}.

  \item Since the color $i$
  appears with color $i+d$, it follows that $p_i=p_{i+d}$ for
  $i=1,\ldots,d$.  Hence $\Pi_\Gam\subseteq T\Pi_{d+1}$.
  It is left to show that any

  \noindent
  $\p=(p_1,\ldots,p_d,p_1,\ldots,p_d,p_{2d+1})\in \Pi_{2d+1}$
  is in $\Pi_\Gam$.  Equivalently, the probability vector
  $\br:=(2p_1,\ldots,2p_d,p_{d+1})$ is the density vector
  of the dimer-monomer covering of $\Z^d$.  For $d=1$, this result is
  straightforward, e.g. \cite{FP}.  So assume that $d>1$.
  Suppose first that all the coordinates of $\br$ are rational and positive:
  $\br=(\frac{i_1}{m},\ldots,\frac{i_d}{m},\frac{i_{d+1}}{m})$,
  where $m$ is a positive integer.  Consider the sequence
  $\m_q=(2qm,2qm,\ldots,2qm)\in \N^d, q\in\N$.  Partition the cube
  $\an{\m_q}$ to $d+1$ boxes with a basis
  $\an{\m_q'}, (2qm,\ldots,2qm) \in \N^{d-1}$:
  $\an{(\m_q',2qi_j)}, j=1,\ldots,d+1$.  Tile the boxes
  $\an{(\m_q',2qi_j)}$ with the dimers in the direction $\e_j$ for
  $j=1,\ldots,d$, and the last box $\an{(\m_q',2qi_{d+1})}$ with
  monomers.  Then $\p=T\br\in \Pi_\Gam$.  Since $\Pi_\Gam$ is closed
  we deduce that $\Pi_\Gam \supseteq T\Pi_{d+1}$.  Hence
   $\Pi_\Gam= T\Pi_{d+1}$.

  In case $\u \in \Der
  P_\Gam$, then $ \nabla P_\Gam(\u)\in \Pi_\Gam= T\Pi_{d+1}$.
  In case $\u \in \partial(P_\Gam(\u)) \setminus \Der
  P_\Gam$, clearly $S(\u) \subset T\Pi_{d+1}$.
  Hence
  $\partial P_\Gam(\u) =
  \conv S(\u)\subset T\Pi_{d+1}$.  Therefore
  $\partial P_{\Gam}(\R^{2d+1}) \subseteq T(\Pi_{d+1})$.

  \item
  According to part (b) of Example \ref{examfr},
  the graph $\Gam$, corresponding to the monomer-dimer model, is friendly, (as explained in \cite[\S4]{FP}).
  Part (\ref{2frnd}) of
  Theorem \ref{frnthm} yields that $h_\Gam$ is concave $\Pi_\Gam=T\Pi_d$.
  {\hspace*{\fill} $\blacksquare$}

   \end{enumerate}
 } 

 \end{Proof}

 Define
  \[R_d(\w):=P_{\Gam}(T_1(\w)).\]
  In analogy with Proposition~\ref{conthgu} and Proposition~\ref{heid},
  $R_d(\w):\R^{d+1}\to \R$ is convex Lipschitz function which
  satisfies the conditions
  \begin{equation}\label{barpdprop1}
   |R_d(\w+\z) - R_d(\w)| \leq \|\z\|_{\max} \qquad  \w,\z \in \R^{d+1},
  \end{equation}
  \begin{equation}\label{barpdprop2}
   R_d(\w) = t + R_d(\w - t\e) \qquad \w \in \R^{d+1}, t\in\R.
  \end{equation}

  We now derive the properties $R_d(\w)$ that are analogous
  to the properties
  of $P_{\Gam}(\u)$ discussed in Section~\ref{densityentropy}.
  First we view $R_d$ as the restriction of $P_{\Gam}$ to
  the $(d+1)$-dimensional subspace $T\R^{d+1}$.
  Observe that $Q\Pi_{2d+1}=QT\Pi_{d+1}=\Pi_{d+1}$.
  Note that the vector $\br=(r_1,\ldots,r_d,r_{d+1})\trans\in \Pi_{d+1}$ can be
  defined intrinsically, where $r_i$ the dimer density in the
  direction $\e_i$ for $i=1,\ldots,d$, and $r_{d+1}$ is the monomer
  density in the lattice $\Z^d$.
 Let

  \[H_d(\br):= h_{\Gam}(T\br), \textrm{ for any } \br\in \Pi_{d+1}.\]
  We view $ H_d(\br)$ as the
  anisotropic dimer-monomer entropy of density $\br$.

  It is
  straightforward to show that $R_d$ satisfies an analogous theorem
  to Theorem \ref{maxchar}. In particular, $R_d(\w) = \max_{\br \in
  \Pi_{d+1}}(\br\trans \w + H_d(\br))$.
  For $\w \in \R^{d+1}$, we denote
  \begin{equation}\label{defpubis}
   \Pi_{d+1}(\w) := \arg \max_{\br \in \Pi_{d+1}} (\br\trans \w + H_d(\br))
   = \{\br \in \Pi_{d+1} : R_d(\w) =  \br\trans \w + H_d(\br)\}.
  \end{equation}
  Because of the equality~(\ref{barpdprop2}), we can use the analogous
  results to Theorems \ref{denthm} and \ref{maxcharhat}. More
  precisely, for $\v \in \R^{d}$, let $P_d(\v)$ be defined as in
  (\ref{hhat}), i.e.,
  \begin{equation}\label{hdv}
  P_d(\v)=R_d(\iota_0(\v))=P_\Gam((\u_0)), \; \u_0\trans = (\v\trans,\v\trans,0),
  \v\trans = (v_1,\ldots,v_d) \in \R^d.
  \end{equation}
 In other words, the two halves of a dimer in the direction of
 $\e_k$ are given the positive weight $x_k=e^{v_k}$ each, and a monomer is
 given the weight $1=e^0$. Then $Z_{\perio}(\m,\v) := Z_{\Gam,\perio}(\m,\u_0)$,
 is the \emph{grand
 partition monomer-dimer (counting) function} in which we sum over
 all monomer-dimer tilings of the torus $T(\m)$, and each tiling
 having exactly $\mu_i$ dimers in the direction $\e_i$ for
 $i=1,\ldots,d$ plus monomers contributes $\prod_{1=1}^d
 e^{2\mu_iv_i}$. As in \cite{FP}, the function $Z(\m,\v) :=
 Z_\Gam(\m,\u_0)$ does not exactly count the weighted monomer-dimer
 covers of $\an{\m}$, because protruding dimers have only half
 of their weight counted.  This can be easily taken care of as in
 \cite{FP}, and the pressure $P_d(\v)$
 is a convex function of $\v \in \R^d$.
 \begin{lemma}\label{invhu}
  Let $\v\trans = (v_1,\ldots,v_d) \in \R^d$ and let $\sigma :
  \an{d} \to \an{d}$ be a permutation.  Then
  $P_d((v_1,\ldots,v_d)\trans) =
  P_d((v_{\sigma(1)},\ldots,v_{\sigma(d)})\trans)$; in other words,
  $P_d(\v)$ is a symmetric function of $v_1,\ldots,v_d$. Similarly
  for $Z(\m,\v)$
 \end{lemma}
 \begin{Proof}  By applying the automorphism of $\N^d$ given by
 \[(m_1,\ldots,m_d) \mapsto (m_{\sigma(1)},\ldots,m_{\sigma(d)})\]
 we obtain the equality
 \begin{equation}\label{invwu}
  Z((m_1,\ldots,m_d),(v_1,\ldots,v_d)\trans) = Z((m_{\sigma(1)},\ldots,m_{\sigma(d)}),
  (v_{\sigma(1)},\ldots,v_{\sigma(d)})\trans),
 \end{equation}
 and the result follows from (\ref{wmufor}).
 \end{Proof}
  Then each for each $\br \in \Pi_{d+1}(\iota_0(\v))$ we have
  $\overline{\br} \in \partial P_d(\v)$. We define $\Delta_d := \overline{\Pi_{d+1}}$
  to be the projection of $\Pi_{d+1}$ on the first $d$ coordinates.  Let
  \begin{equation}\label{defbarhd}
  h_d(\bar \br)=H_d(\br), \quad \br\in\Pi_{d+1}.
  \end{equation}

  We can repeat the proof Theorem \ref{maxcharhat} to obtain:

  \begin{theo}\label{maxcharbpd}
   Let $P_d^*$ be
   the conjugate convex function
   of the pressure function $P_d$.
   Then
   { 
   \renewcommand{\theenumi}{\alph{enumi}}
   \renewcommand{\labelenumi}{(\theenumi)}
   \begin{enumerate}
     \item \label{twoiotaiota} $\bar h_d(\q) \leq - P^*_d(\q)$ for all $\q \in
     \Delta_d$.
     \item \label{fouriotaiota}
      \begin{equation}\label{maxchar1pd}
      P_d(\v) = \max_{\q \in \Delta_{d}}
      ({\q}\trans\v + \bar h_d(\q)) \text{ for all } \v\ \in
      \R^{d}.
      \end{equation}
      For $\v \in \R^{d}$, we denote
      \[\Delta_d(\v) := {\arg \max_{\q \in \Delta_{d}} (
      {\q}\trans \v + \bar h_{d}(\q))},\]
      that is to say $\q(\v) \in \Delta_d(\v)$ if and only if
      \begin{equation}\label{defpd}
      \q(\v) \in \Delta_d \;\;\text{ and }\;\; P_d(\v) =
      {\q(\v)}\trans \v + \bar h_d(\q(\v)).
      \end{equation}
     \item \label{sixiotaiota} $\bar h_d(\q(\v))= -P^*_d(\q(\v))$.
     \item \label{fiveiotaiota} $\Delta_d(\v) \subseteq \partial P_{d}(\v)$. In
     particular, if $\v \in \Der P_d$, then $\Delta_d(\v) = \{\nabla
     P_d(\v)\}$. Therefore $\partial P_d(\Der P_d) \subseteq
     \Delta_d$.
     \item \label{threebisiotaiota}
     Let $\v \in \R^{d} \setminus \Der P_d$, and let $S(\v)$
     consist of all the limits of sequences $\nabla P_d(\v_i)$
     such that $\v_i \in \Der P_d$ and $\v_i \to \v$. Then $S(\v)
     \subseteq \Delta_d(\v)$.
     \item \label{oneiotaiota} $\conv \Delta_d = \dom P^*_d$.
   \end{enumerate}
   } 
  \end{theo}

 Thus, the first order phase transition occurs at the
 points $\v$ where $P_d$ is not differentiable.

  As in \cite{Ha1,Bax1,HL} we consider the total dimer density $q:=q_1 +
  \cdots + q_d$.  This is equivalent to the equalities $v_1 = \cdots = v_d = v = \log
  s$, where $s>0$ is the weight of a half a dimer in any direction.
  We define
  $\pres_d(v):=P_d((v,\ldots,v)\trans)=P_d(v\e):\R\to \R$.
  Then $\pres_d$ is a nondecreasing convex Lipschitz
  function satisfies $|\pres_d(u)-\pres_d(v)|\le |u-v|$.

  \begin{prop}\label{totdimden}  For each $d\in\N$
  $Q_d(\Delta_d)=[0,1]$.  Let
  \begin{equation}\label{fdefHTp}
  \hat h_d(p):=\max_{\q\in\Delta_d, Q_d\q=p}
  \bar h_d(\q), \textrm{ for each }p\in[0,1].
  \end{equation}
  Then
  \begin{equation}\label{maxcharPT}
  \pres_d(v)=\max_{p\in [0,1]} pv + \hat h_d(p).
  \end{equation}
  Furthermore, $\hat h_d(p)$ is the $p$-dimer
  entropy as defined in \cite{Ha1} or \cite{FP}.

  \end{prop}

  \begin{Proof}  Let $p\in[0,1]$ be the limit density of dimers, abbreviated here
  as $p$-dimer density, as $\lan
  \m\ran\to\infty$ as discussed in \cite{FP}.
  We recall the definition of the  $p$-dimer density in terms of
  quantities defined in \S\ref{densityentropy}.  (See in particular
  Definition \ref{deflamb1}.)

  For each $\m\in \N^d$ and a nonnegative integer $a\in [0,\vol(\m)]$ define
  \begin{equation}\label{defCgama}
  C_\Gam(\an{\m},a):=\cup_{\bc=(c_1,\ldots,c_{2d+1})\in \Pi_{2d+1}(\vol(\m)),
  c_{2d+1}=a} C_\Gam(\an{\m},\bc).
  \end{equation}
  So $C_\Gam(\an{\m},a)$ is roughly equal to the set of all
  covering of the box $\an{\m}\subset \Z^d$ with monomer-dimers,
  such that the number of monomers is $a$.  (It may happen that some
  of the dimers protruding "out of" the box $\an{\m}$, see
  \cite{FP}.)  Then $p\in [0,1]$ is dimer density
  if there exists a sequence of boxes $\an{\m_q}\subset \N^d$ and a
  corresponding sequence of nonnegative integers $a_q\in
  [0,\vol(\m_q)]$, such that
  \begin{equation}\label{defplimden}
  \m_q\to\infty, \; C_\Gam(\an{\m_q},a_q)\ne \emptyset\; \forall q\in
  \N, \textrm{ and }
  \lim_{q\to\infty} \frac{a_q}{\vol(\m_q)}=1-p.
  \end{equation}
  From the definition
  of the density set $\Delta_{d}$ of the dimers it follows
  that $p$ is a dimer density if and only $p=Q_d\q$ for some
  $\q\in \Delta_{d}$.  Since $\Delta_d=\overline{\Pi_{d+1}}$,
  it follows that $Q_d(\Delta_d)=[0,1]$.

  For each $p\in [0,1]$ let
  \begin{equation}\label{defhdp}
  h_d(p):=\sup_{\m_q,a_q} \limsup_{q\to\infty}
  \frac{\log\#C_\Gam(\an{\m_q},a_q)}{\vol(\m_q)}\ge 0,
  \end{equation}
  where the supremum is taken over all the sequences satisfying
  (\ref{defplimden}).
  Then $h_d(p)$ is the $p$-dimer entropy as defined in \cite{FP}.
  Let $\hpres_d(p)$ be defined as in (\ref{fdefHTp}).
  We claim
  \begin{equation}\label{hd=HT}
  h_d(p)=\hat h_d(p) \textrm{ for all } p\in [0,1].
  \end{equation}

  Observe first that $C_\Gam(\an{\m},\bc)\subseteq
  C_\Gam(\an{\m},c_{2d+1})$ for any $\bc=(c_1,\ldots,c_{2d+1})\in
  \Pi_{2d+1}(\vol(\m))$.  The definition of $\bar h_d(\q)$ and $h_d(p)$
  implies straightforward the inequality $\bar h_d(\q)\le
  h_d(Q_d\q)$.  Hence $\hat h_d(p)\le h_d(p)$.
  (\ref{defCgama}) yields the inequality
  \begin{eqnarray*}
  &&\# C_\Gam(\an{\m},a) \leq \binom{\vol(\m) + 2d}{2d} \times \\
  &&\max_ {\bc=(c_1,\ldots,c_{2d+1})\in \Pi_{2d+1}(\vol(\m)),
  c_{2d+1}=a} \#C_\Gam(\an{\m},\bc).
  \end{eqnarray*}
  Use the arguments of the proof of part (b) Theorem \ref{maxchar}
  to deduce the existence of $\q\in \Delta_d$, such that $Q_d\q=p$
  and $h_d(p)\le \bar h_d(\q)$.  Hence $h_d(p)\le \hat h_d(p)$ and
  therefore $h_d(p)=\hat h_d(p)$.

  To show (\ref{maxcharPT}), take $\v = v\e$ in (\ref{maxchar1pd}).
  We have ${\q}\trans \v = pv$, where $p = {\q}\trans
  \e$. In (\ref{maxchar1pd}) take the maximum in two stages. The
  first stage is for fixed $p$, and the second stage over all $p$.
  \end{Proof}


  The results of \cite{HL} yield that
  $\pres_{d}(v)$ is analytic.
  Since $\pres_{d}(v)$ is also convex and not affine it follows that $\pres_d'(v)$ can
  not be constant on any interval $(a,b)$.  Hence $p(v):=\pres_d'(v)$ is
  increasing on $\R$ with $p(-\infty)=0$ (no dimers) and
  $p(\infty)=1$ (only dimers).  Therefore the analytic function
  $p:\R\to (0,1)$ has an increasing analytic inverse $v(p):(0,1)\to \R$.
  Recall that $\pres_d^*(p)$ is a convex function of $p$.
  Moreover
  \[\frac{d \pres_d^*}{dp}=v(p) + p\frac{dv(p)}{dp} -
  \frac{d\pres_d}{dv}\frac{dv}{dp}=v(p)+p\frac{dv(p)}{dp}-p\frac{dv(p)}{dp}=v(p).\]
  As $v(p)$ is an increasing function of $p$ it follows that
  $\pres_d^*(p)$ is a strictly convex function on $(0,1)$.
  The corresponding dimer density entropy $\pres_d^*(p)=-
  h_d(p)$ is a strictly concave function.  This is an improvement
  of the result of Hammersley \cite{Ha1} which showed that
  $h_d(p)$ is a concave function on $(0,1)$.
  \cite[Corollary 3.2]{FKLM} claims a stronger result, namely
  $h_d(p)+\frac{1}{2}(p\log p +(1-p)\log(1-p))$ is a concave
  function on $[0,1]$.  (Observe that $p\log p +(1-p)\log(1-p)$
  is a strict convex function on $[0,1]$.)

  Since $\pres_d$ is differentiable it follows
  that $\pres_d^*(p)=pv(p)-\pres_d(v(p))$.  Hence we obtain the well
  known formula, e.g. \cite{Bax1}
  \begin{equation}\label{stadforhdp}
  h_d(p(v))=\pres_d(v)-p(v)v,  \textrm{ where }p(v)=\pres_d'(v) \textrm{ for all
  } v\in \R.
  \end{equation}

  Note that
  $h_d(0):=\lim_{p\searrow 0} h_d(p)=0$ and
  $h_d(1):=\lim_{p\nearrow 1} h_d(p)$ is the $d$-dimensional
  dimer-entropy.

 \section{Symmetric encoding of the monomer-dimer model}\label{sec:symencod}

 The disadvantage of the encoding (\ref{encoding}) is that the
 $\Gamma_k$ are not symmetric, so we cannot apply the results of
 Section \ref{sec:SymmetricNNSOFT} directly. However, as pointed
 out in \cite{FP}, there is a hidden symmetry, which enables us to
 obtain results analogous to those of Section
 \ref{sec:SymmetricNNSOFT}.  We now adapt the arguments of
 \cite[Section 6]{FP} to $P_d(\v)$, the pressure corresponding to the
 weighted monomer-dimer coverings.

 For $d \in \N$, $K \subseteq \an{d}$ and $\m \in \N^d$, we denote
 by $\an{\m_K}$ the projection of $\an{\m}$ on the coordinates with
 indices in $K$. Let $C_{\perio,K}(\m)$ be the set of monomer-dimer
 covers of $T(\m_K) \times \an{\m_{\an{d} \setminus K}}$, and
 $Z_{\perio,K}(\m,\v)$ the corresponding weighted sum. Thus
 $C_{\perio,\an{d}}(\m) = C_{\perio}(\m)$ and
 $Z_{\perio,\an{d}}(\m,\v) = Z_{\perio}(\m,\v)$. Note that by the
 isotropy of our $\Gam$, $\#C_{\perio,K}(\m)$ is invariant under
 permutations of the components of $\m$ if $K$ undergoes a
 corresponding change. Similarly for $Z_{\perio,K}(\m,\v)$, if $K$
 and $\v$ undergo a corresponding change.

 In order to analyze $C_{\perio,\{d\}}(\m)$, we focus on the
 dimers in the cover lying along the direction $\e_d$. More
 precisely, with $\m' = (m_1,\ldots,m_{d-1})$, we consider
 $\an{\m'} \times T(m_d)$ as consisting of $m_d$ levels isomorphic
 to $\an{\m'}$. A subset $S$ of the sites on level $q$ is covered
 by dimers joining levels $q-1$ and $q$ (with level $0$ understood
 as level $m_d$); a subset $T$ disjoint from $S$ is covered by
 dimers joining levels $q$ and $q+1$ (with level $m_d + 1$
 understood as level $1$); and the remainder $U$ of level $q$ is
 covered by monomers and dimers lying entirely within level $q$.
 We are interested in counting the coverings of $U$ subject to
 various restrictions. With that in mind, for $\m' \in \N^{d-1}$
 we define an undirected graph $G(\m')$ whose vertices are the
 subsets of $\an{\m'}$, in which subsets $S$ and $T$ are adjacent
 if and only if $S \cap T = \emptyset$. When $S \cap T =
 \emptyset$ we also define, using $U = \an{\m'} \setminus (S \cup
 T)$, and $\v' = (v_1,\ldots,v_{d-1})\trans$,
 \begin{align*}
   \tilde{a}_{ST}(\v') &= \text{ sum of weighted monomer-dimer tilings of } U
   \\
   \tilde{b}_{ST}(\v') &= \text{ sum of weighted monomer-dimer tilings of }
   U \text{ viewed as a subset of }
   \\ T(\m')
   \\
   \tilde{p}_{ST}(\v') &=
   \begin{aligned}[t]
   &\text{ sum of weighted monomer-dimer covers of } U, \text{ viewed as a subset of }
   \\
   &T(m_1) \times \an{(m_2,\ldots,m_{d-1})},
   \text{ each monomer within } U, \text{ and each}
   \\
   &\text{ dimer meeting } U \text{ but not } S \cup T.
   \end{aligned}
   \\
   \tilde{c}_{ST}(\v') &=
   \begin{aligned}[t]
   &\text{ sum of weighted monomer-dimer covers of } U, \text{ each monomer
   within } U,
   \\
   &\text{ and each dimer meeting } U \text{ but not } S \cup T.
   \end{aligned}
 \end{align*}
 In the tilings/covers counted by $\tilde{a}_{ST}(\v')$,
 $\tilde{b}_{ST}(\v')$, $\tilde{p}_{ST}(\v')$,
 $\tilde{c}_{ST}(\v')$, each monomer lies within $U$ and each dimer
 meets $U$ but not $S \cup T$. In $\tilde{a}_{ST}(\v')$, each dimer
 occupies two sites of $U$ that are adjacent in $\an{\m'}$. In
 $\tilde{b}_{ST}(\v')$, each dimer occupies two sites of $U$ that
 are adjacent in $T(\m')$, so is allowed to ``wrap around''. In
 $\tilde{p}_{ST}(\v')$, the dimers in the direction of $\e_1$ are
 allowed to ``wrap around'' and the other dimers are allowed to
 ``protrude out'' of $\an{(m_2,\ldots,m_{d-1})}$. In
 $\tilde{c}_{ST}(\v')$, the dimers may ``protrude'' out of
 $\an{\m'}$.  The weight of each monomer-dimer cover is a product of
 the weights of dimers and ``half'' dimers appearing in the cover.
 If a dimer in the direction of $\e_k$ is entirely within $U$, then
 its weight is $e^{2v_k}$.  If a dimer ``protrudes out'' in the
 direction of $\e_k$, then its weight is only $e^{v_k}$. Therefore
 \[\tilde{a}_{ST}(\v') \leq \tilde{b}_{ST}(\v')
 \leq \tilde{p}_{ST}(\v') \leq \tilde{c}_{ST}(\v').\] By
 definition, if $U = \emptyset$, then $\tilde{a}_{ST}(\v') =
 \tilde{b}_{ST}(\v') = \tilde{p}_{ST}(\v') = \tilde{c}_{ST}(\v') =
 1$. Notice that when $d=2$, there is no distinction between
 $\tilde{b}_{ST}(\v')$ and $\tilde{p}_{ST}(\v')$.

 We define matrices $A(\m',\v) = (a_{ST}(\v))_{S,T \subseteq
 \an{\m'}}$, $B(\m',\v) = (b_{ST}(\v))_{S,T \subseteq \an{\m'}}$,
 $P(\m',\v) = (p_{ST}(\v))_{S,T \subseteq \an{\m'}}$, $C(\m',\v) =
 (c_{ST}(\v))_{S,T \subseteq \an{\m'}}$ with rows and columns
 indexed by subsets of $\an{\m'}$ as follows:

 \begin{align*}
   &\twocases{A(\m',\v)_{ST}}{\tilde{a}_{ST}(\v')e^{(\#S+\#T)v_d}}
   {S \cap T = \emptyset}0{S \cap T \neq \emptyset}
   \\
   &\twocases{B(\m',\v)_{ST}}{\tilde{b}_{ST}(\v')e^{(\#S+\#T)v_d}}
   {S \cap T = \emptyset}0{S \cap T \neq \emptyset}
   \\
   &\twocases{P(\m',\v)_{ST}}{\tilde{p}_{ST}(\v')e^{(\#S+\#T)v_d}}{S \cap T =
   \emptyset}0{S \cap T \neq \emptyset}
   \\
   &\twocases{C(\m',\v)_{ST}}{\tilde{c}_{ST}(\v')e^{(\#S+\#T)v_d}}{S \cap T =
   \emptyset}0 {S \cap T \neq \emptyset.}
 \end{align*}
 Thus $A(\m',\v)$, $B(\m',\v)$, $P(\m',\v)$, $C(\m',\v)$ are symmetric matrices---here is
 the ``hidden symmetry'' referred to above.  Clearly
 \[0 \leq A(\m',\v) \leq B(\m',\v) \leq P(\m',\v) \leq C(\m',\v)\]
 (where the inequalities indicate componentwise comparisons). We
 use the notation $\alpha(\m',\v)$, $\beta(\m',\v)$, $\pi(\m',\v)$,
 $\gamma(\m',\v)$ for the spectral radii of these matrices,
 respectively, so that
 \[\alpha(\m',\v) \leq \beta(\m',\v) \leq \pi(\m',\v) \leq \gamma(\m',\v).\]
 Note that by Kingman's theorem \cite{Kin} all the
 spectral radii are log-convex in $\v$.

 The four matrices have the same zero-nonzero pattern, namely
 the adjacency matrix of the graph $G(\m')$. If the graph is
 connected, we say that the matrix is \emph{irreducible}; if in
 addition the greatest common divisor of the lengths of all its
 cycles is 1, equivalently for sufficiently high powers of the
 matrix all entries are strictly positive, we say that the matrix
 is \emph{primitive}.
 \begin{prop}\label{symprop}
 Let $2 \leq  d\in \N$ and $\m = (\m',m_d) \in \N^d$. Then

 \begin{enumerate}
  \item[(a)] $\tr A(\m',\v)^{m_d}$ is the sum of the weighted monomer-dimer
  tilings of $\an{\m'} \times T(m_d)$ ;
  \item[(b)] $\tr B(\m',\v)^{m_d} =
  Z_{\perio}(\m,\v)$;
  \item[(c)] $\tr P(\m',\v)^{m_d} =
  Z_{\perio,\{1,d\}}(\m,\v)$;
  \item [(d)] $\tr C(\m',\v)^{m_d} =
  Z_{\perio,\{d\}}(\m,\v)$ ;
  \item[(e)] for $m_d \geq 2$, if column vector $\x(\v) = (x_S(\v))_{S \subseteq \an{\m'}}$
  is given by $x_S(\v) = \tilde{b}_{S\emptyset}(\v')e^{\#S v_d}$, then
  $\x(\v)^{\trans} B(\m',\v)^{m_d - 2} \x(\v) = Z_{\perio,\an{d-1}}(\m,\v)$,

  if column vector $\y(\v)= (y_S(\v))_{S \subseteq \an{\m'}}$ is given by $y_S(\v) =
  \tilde{c}_{S\emptyset}(\v')e^{\#S v_d}$,

  \noindent
  then
  $\y(\v)^{\trans} C(\m',\v)^{m_d - 2} \y(\v) = Z(\m,\v)$,

  and if column vector $\z(\v) = (z_S(\v))_{S \subseteq \an{\m'}}$
  is given by $z_S(\v) = \tilde{p}_{S\emptyset}(\v')e^{\#S v_d}$, then
  $\z(\v)^{\trans} P(\m',\v)^{m_d - 2} \z(\v) =
  Z_{\perio,\{1\}}(\m,\v)$;
  \item[(f)] the matrices $A(\m',\v)$, $B(\m',\v)$, $P(\m',\v)$, $C(\m',\v)$ are
  primitive.
 \end{enumerate}
 \end{prop}
 \begin{Proof} We begin with proving (b), observing that (a), (c), (d) and
 (e) are similar. Assume first that $m_d = 1$, and let $\phi \in
 C_{\perio}(\m)$. Since $\phi$ can be extended periodically in the
 direction of $\e_d$ with period 1, it can be viewed as an element
 of $C_{\perio}(\m')$. Therefore $\# C_{\perio}(\m) = \#
 C_{\perio}(\m')$ and moreover, $Z_{\perio}(\m,\v) =
 Z_{\perio}(\m',\v')$ ($v_d$ does not matter since no dimer lies in
 the direction of $\e_d$). We have $\tr B(\m',\v) = \sum_{S
 \subseteq \an{\m'}} b_{SS}(\v)$. Only the term $S = \emptyset$
 contributes to the sum, and for this term we have $U = \an{\m'}$
 and $b_{\emptyset\emptyset} = Z_{\perio}(\m',\v') =
 Z_{\perio}(\m,\v)$. Hence $\tr B(\m',\v) = Z_{\perio}(\m,\v)$. Now
 assume that $m_d > 1$, and consider a closed walk
 $S_1,S_2,\ldots,S_{m_d},S_1$ of length $m_d$ in $G(\m')$. For each
 $\p' \in S_q$ place a dimer in the direction of $\e_d$ occupying
 the sites $(\p',q)$ and $(\p',q+1)$ (with $m_d + 1$ wrapping around
 to $1$). We want to extend these dimers to a monomer-dimer tiling
 of $T(\m') \times T(m_d) = T(\m)$, i.e., to a member of
 $C_{\perio}(\m)$, by monomers and by dimers not in the direction of
 $\e_d$, i.e., lying within the levels $1,\ldots,m_d$. The weighted
 number of choices of such monomers and dimers to fill the remainder
 of level $q$ is given by $\tilde{b}_{s_{q-1}S_q}(\v')$, and
 together with the weight of the dimers in the direction of $\e_d$
 intersecting level $q$ it becomes $b_{S_{q-1}S_q}(\v)$. Therefore
 the weighted number of extensions to a member of $C_{\perio}(\m)$,
 i.e., the corresponding term of $Z_{\perio}(\m,\v)$, is
 $b_{S_{1}S_{2}}(\v) b_{S_{2}S_{3}}(\v) \cdots b_{S_{m_d -
 1}S_{m_d}}(\v) b_{S_{m_d}S_{1}}(\v)$. Conversely, each term of
 $Z_{\perio}(\m,\v)$ is obtained in this way. Hence
 $Z_{\perio}(\m,\v)$ is the sum of all the products of the above
 form, namely $\tr B(\m',\v)^{m_d}$.

 To prove (f), we note that $A(\m',\v)$ is irreducible, since
 whenever $S \cap T = \emptyset$, $U$ can be tiled by monomers and
 therefore each subset of $\an{\m'}$ is adjacent to $\emptyset$ in
 $G(\m')$. Furthermore, $A(\m',\v)$ is primitive since the graph has
 a cycle of length 1 from $\emptyset$ to $\emptyset$. Since
 $A(\m',\v) \leq B(\m',\v) \leq P(\m',\v) \leq C(\m',\v)$, it
 follows that $B(\m',\v)$, $P(\m'.\v)$ and $C(\m',\v)$ are also
 primitive. \end{Proof}

 For the next lemma, we define $C_0(\m)$ as the set of colorings
 of $\an{\m}$ corresponding to its monomer-dimer tilings (so no dimer
 protrudes out of $\an{\m}$), and the corresponding weighted sum
 \[Z_0(\m,\v) = \sum_{\phi \in C_0(\m)} e^{\bc(\phi)\trans \u}
 \qquad \u\trans = (\v\trans,\v\trans,0).\]
 \begin{lemma}\label{albetgain}
  Let $2 \leq  d\in \N$ and $\m' \in \N^{d-1}, \v\in \R^d$.  Then
  \begin{align}
   &\lim_{m_d \to \infty} \frac{\log Z_0((\m',m_d),\v)}{m_d} =
   \log \alpha(\m',\v)\label{albetgaina}
   \\
   &\lim_{m_d \to \infty}
   \frac{\log Z_{\perio,\an{d-1}}((\m',m_d),\v)}{m_d}=
   \log \beta(\m',\v)\label{albetgainb}
   \\
   &\lim_{m_d \to \infty}
   \frac{\log Z_{\perio,\{1\}}((\m',m_d),\v)}{m_d}=
   \log \pi(\m',\v)\label{albetgainbprime}
   \\
   &\lim_{m_d \to \infty} \frac{\log Z((\m',m_d),\v)}{m_d}=
   \log \gamma(\m',\v)\label{albetgainc}
   \end{align}
  \end{lemma}
  \begin{Proof} From Part (a) of Proposition \ref{symprop}
  $Z_0((\m',m_d),\v) \leq \tr A(\m',\v)^{m_d}$, and therefore
  \begin{equation}\label{limsuplelog}
   \limsup_{m_d \to \infty} \frac{\log Z_0((\m',m_d),\v)}{m_d}
   \leq  \limsup_{m_d \to \infty} \frac{\log \tr A(\m',\v)^{m_d}}{m_d} =
   \log \alpha(\m',\v).
  \end{equation}
  The equality in (\ref{limsuplelog}) follows from a
  characterization of $\rho(M)$ for a square matrix $M \geq 0$,
  namely $\rho(M) = \limsup_{k \to \infty} (\tr
  M^k)^{\frac{1}{k}}$ (see for example Proposition 10.3 of
  \cite{Fr2}). Since $-\log Z_0(\m',m_d)$ is subadditive in $m_d$,
  the first $\limsup$ in (\ref{limsuplelog}) can be replaced by a
  $\lim$. In order to prove the reverse inequality and thus
  (\ref{albetgaina}), observe that each monomer-dimer tiling of
  $\an{\m'} \times T(m_d)$ extends to a monomer-dimer tiling in
  $C_0(\m',m_d + 1)$ having the same weight (replace each dimer
  occupying $(\m',1)$ and $(\m',m_d)$ by a monomer occupying
  $(\m',1)$ and a dimer occupying $(\m',m_d)$ and $(\m',m_d + 1)$,
  and tile the rest with monomers). Hence $Z_0((\m',m_d + 1),\v)
  \geq \tr A(\m',\v)^{m_d}$ by Part (a) of Proposition
  \ref{symprop}. Therefore, since $-\log Z_0((\m',m_d),\v)$ is
  subadditive in $m_d$ and thus the limits below exist, we obtain
  \begin{multline*}
   \lim_{m_d \to \infty} \frac{\log Z_0((\m',m_d),\v)}{m_d}
   = \lim_{m_d \to \infty} \frac{\log Z_0((\m',m_d + 1),\v)}{m_d}
   \\
   \geq \limsup_{m_d \to \infty} \frac{\log \tr A(\m',\v)^{m_d}}{m_d}
   = \log \alpha(\m',\v).
  \end{multline*}

  To prove (\ref{albetgainb}), (\ref{albetgainbprime}),
  (\ref{albetgainc}), we use the fact mentioned in the proof of
  Proposition~\ref{onedimpressure} that if $M \geq 0$ and $\w$ is
  a column vector with positive entries, then $\rho(M) = \lim_{k
  \to \infty} (\w^{\trans} M^k \w)^{\frac{1}{k}}$. Applying this
  to $M = B(\m',\v), P(\m',\v), C(\m',\v)$ and using Part (e) of
  Proposition \ref{symprop} with $\w = \x(\v),\z(\v),\y(\v)$
  defined there proves (\ref{albetgainb}), \ref{albetgainbprime}),
  (\ref{albetgainc}). \end{Proof}

  Now we introduce the following notation. For $\m \in \N^d$ and
  $k \in \an{d}$,

  \noindent
  $\m^{\sim k} := (m_1,\ldots,m_{k-1},
  m_{k+1},\ldots,m_d)\in \N^{d-1}$. As special cases we have the
  previous notation $\m' = \m^{\sim d}$ and $\m^- = \m^{\sim 1}$.
  For $\v=(v_1,\ldots,v_d)\trans\in\R^d$ we use the notation
  $\v^k:=(v_1,\ldots,v_{k-1},v_{k+1},\ldots,v_d,v_k)\trans$.
  Note that $\v^d=\v$.
  Part (b) of Proposition \ref{symprop} implies
  \begin{equation}\label{permid}
  Z_{\perio}(\m,\v) = \tr B(\m',\v)^{m_d} = \tr B(\m^{\sim k},\v^k)^{m_k}.
  \end{equation}
  \begin{prop}\label{perest}  Let $\m \in \N^d, \v\in\R^d$,
  and assume that $m_d$ is even.  Then each
  $k \in \an{d-1}$ satisfies
  \begin{equation}
  \frac{\log \beta(\m^{\sim d},\v)}{\vol(\m^{\sim d})} \leq
  \frac{\log 2}{m_k} + \frac{\log \beta(\m^{\sim k},\v^k)}{{\vol(\m^{\sim k})}}
  \label{perest1}.
  \end{equation}
  \end{prop}
  \begin{Proof} We have
  \begin{equation*}
   \beta(\m^{\sim d},\v)^{m_d} \leq \tr B(\m^{\sim d},\v)^{m_d}
   = \tr B(\m^{\sim k},\v^k)^{m_k}
   \leq 2^{\vol(\m^{\sim k})} \beta(\m^{\sim k},\v^k)^{m_k}.
  \end{equation*}
  The first inequality above follows since $\beta(\m^{\sim d},\v)$ is
  one of the eigenvalues of $B(\m^{\sim d},\v)$, which are all real,
  and $m_d$ is even; the next equality from (\ref{permid});
  and the last inequality from the fact that
  $B(\m^{\sim k},\v^k)$ has $2^{\vol(\m^{\sim k})}$ eigenvalues, all
  real, whose moduli are at most $\beta(\m^{\sim k},\v^k)$.
  Taking logarithms and dividing by $\vol(\m)$, we deduce
  (\ref{perest1}).
  \end{Proof}
  We define
  \begin{align}
   \overline{P}_{d-1}(m_1,\v) &:= \lim_{\m^{-} \to \infty}
   \frac{\log Z_{\perio,\{1\}}((m_1,\m^{-}),\v)}
   {\vol(\m^{-})}, \qquad m_1 \in \N \label{hbardef1}
   \\
   \overline{P}_{d-1}(0,\v) &:= \log 2 \label{hbardef2}.
  \end{align}
  Notice that for $m_1 \in \N$, $\overline{P}_{d-1}(m_1,\v)$ is
  the same as $\overline{P}_{\Gam}(m_1,\u)$ defined in
  (\ref{defbarh}), where $\Gam$ is given by (\ref{encoding}). For
  this reason the limit $\overline{P}_{d-1}(m_1,\v)$ exists. The
  following theorem is an analog of Theorem \ref{perulb} and
  (\ref{entdef}).
  \begin{theo}\label{main}  Let $2 \leq  d \in \N$, $p, r \in \N$,
  $q \in \Z_+$, $\v \in \R^d$.  Then
   \begin{equation}
    \frac{\overline{P}_{d-1}(2r,\v)}{2r} \geq  P_d(\v) \geq
    \frac{\overline{P}_{d-1}(p+2q,\v)-\overline{P}_{d-1}(2q,\v)}{p}.
    \label{imlb}
   \end{equation}
   Let $\m'=(m_1,\ldots,m_{d-1}) \in \N^{d-1}$ and assume that $m_1,\ldots,m_{d-1}$
   are even.  Then
   \begin{equation}
    P_d(\v) \leq \frac{\log \beta(\m',\v)}{\vol(\m')}. \label{imub1}
   \end{equation}
  \end{theo}
 \begin{Proof}  Since $\#C_0(\m + 2\1)\geq \#C(\m)$ as explained in
 \cite{FP}, it follows that $Z_0(\m + 2\1,\v) \geq Z(\m,\v)$.
 Hence, as in \cite[formula~(4.6) and (6.19)]{FP} and by
 Lamma~\ref{albetgain},
  \begin{equation}\label{stcharhmd}
   P_d(\v) = \lim_{\m' \to \infty} \frac{\log \alpha(\m',\v)}{\vol(\m')}
       = \lim_{\m' \to \infty} \frac{\log \beta(\m',\v)}{\vol(\m')}
       = \lim_{\m' \to \infty} \frac{\log \gamma(\m',\v)}{\vol(\m')}.
  \end{equation}
  First we prove (\ref{imub1}).
  Let $\m' = (m_1,\ldots,m_{d-1}) \in \N^{d-1}$, $m_1,\ldots,m_{d-1}$
  even, and let $\s = (s_1,\ldots,s_{d-1}) \in \N^{d-1}$ be arbitrary.  Set
  \begin{multline*}
   \m_1 = (s_1,\ldots,s_{d-1},m_1),\quad \m_2 = (s_2,\ldots,s_{d},m_1,m_2),\quad \ldots,
   \\
   \m_{d-1} = (s_{d},m_1,\ldots,m_{d-1}).
  \end{multline*}
  Note that (\ref{perest1}) with $k = 1$ states that
  \[\frac{\log \beta(\m',\v)}{\vol(\m')} \leq
  \frac{\log 2}{m_1} + \frac{\log \beta(\m^-,\v^1)}{\vol(\m^{-})}.\]
  Using it $d-1$ times along with $\s = \m_1'$, $\m_1^- = \m_2'$,
  $\m_2^- =  \m_3'$, etc., we obtain
  \begin{multline*}
   \frac{\log \beta(\s,(v_2,v_3,\ldots,v_d,v_1)\trans)}{\vol(\s)} \leq
   \frac{\log 2}{s_1} + \frac{\log \beta(\m_1^-,(v_3,\ldots,v_d,v_1,v_2)\trans)}{\vol(\m_1^-)}
   \leq \\
   \frac{\log 2}{s_1} + \frac{\log 2}{s_2} +
   \frac{\log\beta(\m_2^-), (v_4,\ldots,v_d,v_1,v_2,v_3)\trans}{\vol(\m_2^-)}
   \leq \cdots
   \\
   \leq \sum_{j=1}^{d-1}\frac{\log 2}{s_j} + \frac{\log \beta(\m',\v)}{\vol(\m')}.
  \end{multline*}
  Letting $\s \to\infty$ and using (\ref{stcharhmd}) and Lemma \ref{invhu} for the
  left-hand side, we deduce (\ref{imub1}).

  We now demonstrate the lower bound in (\ref{imlb}). Let $\m^-
  \in \N^{d-1}$, $s \in \N$, $q\in \Z_+$. Assume first that $q \in
  \N$. Since $\gamma(\m^-,\v^1) = \rho(C(\m^-,\v^1))$ and
  $C(\m^-,\v^1)$ is symmetric, it follows as in the arguments for
  (\ref{rhogeqtraceratio}) and by the analog of (\ref{permid}) for
  $C(\m^-,\v^1)$ that
  \begin{equation}\label{lbhd}
   \gamma(\m^-,\v^1)^s \geq \frac{\tr C(\m^-,\v^1)^{s+2q}}{\tr C(\m^-,\v^1)^{2q}} =
   \frac{Z_{\perio,\{1\}}(s+2q,\m^-,\v)}{Z_{\perio,\{1\}}(2q,\m^-),\v)}.
  \end{equation}
  Taking logarithms, dividing by $\vol(\m^-)$, letting $\m^- \to
  \infty$, and using (\ref{stcharhmd}), Lemma \ref{invhu} and the definition of
  $\overline{P}_{d-1}(m_1,\v)$, we deduce the lower bound in (\ref{imlb}) for the
  case $q \in \N$. If $q = 0$, we have to replace the denominators
  in (\ref{lbhd}) by $\tr I = 2^{\vol(\m^-)}$, and the lower
  bound in (\ref{imlb}) is verified by (\ref{hbardef2}).

  We now prove the upper bound of (\ref{imlb}). Let
  $\v^1=(v_2,\ldots,v_{d-1},v_1)\trans$. For each $\m' \in \N^{d-1}$
  we have
  \[\gamma(\m',\v^1)^{2r} \leq \tr C(\m',\v^1)^{2r} = Z_{\perio,\{d\}}
  ((\m',2r),\v^1) = Z_{\perio,\{1\}}((2r,\m'),\v),\]
  where the inequality above is true because the eigenvalues of the
  symmetric matrix $C(\m',\v^1)$ are real and $\gamma(\m',\v^1)$ is
  one of them, the first equality follows from Part (d) of
  Proposition \ref{symprop}, and the last equality from
  (\ref{invwu}). Therefore
  \[\frac{\log \gamma(\m',\v^1)}{\vol(\m')} \leq
  \frac{\log Z_{\perio,\{1\}}((2r,\m'),\v)}{2r\vol(\m')},\]
  and letting $\m' \to \infty$, we deduce the upper bound of
  (\ref{imlb}) by (\ref{stcharhmd}), Lemma \ref{invhu} and the
  definition of $\overline{P}_{d-1}(m_1,\v)$.  \end{Proof}
 In view of (\ref{stcharhmd}) we assume that
  $\frac{\log
  \beta(\m',\v)}{\vol(\m')}$ is a good approximation to
  $P_d(\v)$, and its partial derivative
  $\frac{1}{\vol(\m')\beta(\m',\v)} \frac{\partial \beta(\m',\v)}{\partial v_i}$
  is a good approximation to $q_i := \frac{\partial P_d(\v)}{\partial v_i}$, the density of dimers in
  the direction of $\e_i$.

 \section{Numerical computations for the monomer-dimer model in
 $\Z^2$}\label{sec:graphs}

 In this section we explain in detail our computations for two
 dimensional pressure $P_2(\v)=P_2(v_1,v_2)$ along the lines outlined in
 Sections~\ref{sec:weightedmonomerdimertiligs}-\ref{sec:symencod}.  Our computations
 based on our ability to compute the spectral radius of the
 transfer matrix corresponding to the monomer dimer tiling of
 the torus $(\Z/m)\times \Z$.
 This is a two dimensional integer
 lattice corresponding to a circle of circumference $m$ times
 the real line. In the notation of Section~\ref{sec:symencod}
 this lattice is given by $T(m)\times \Z$.  This
 transfer matrix is denoted by $B(m,\v)$.  Let
 $x=e^{v_1},y=e^{v_2}$.  The the weight of the dimer in
 direction $X$, i.e. the horizontal dimer that lies entirely on the circle
 $T(m)$, is $x^2$.  The weight of the dimer in the direction
 $Y$, i.e. the vertical dimer that lies on two adjacent circles, is
 $y^2$.  The matrix $B(m,\v)$ is of order $2^m$, corresponding
 to all subsets of $\an{m}$.  Denote by $2^{\an{m}}$ the set of
 all subsets of $\an{m}$.  For $S\in 2^{\an{m}}$ denote by
 $\#S$ the cardinality of the set $S$.
 Then $B(m,\v)=[y^{\#S+\#T}f(x,S,T)]_{S,T\in 2^{\an{m}}}$.
 Here $f(x,S,T)=0$ if $S\cap T\ne \emptyset$.  For $S\cap
 T=\emptyset$ the function $f(x,S,T)$ is a polynomial in $x$,
 which is the sum of the following monomials.  Consider the set
 $F:=\an{m}\backslash{S\cup T}$ viewed as a subset of the torus
 $T(m)$.  Let $\cF$ be a tiling of $F$ with monomers and
 dimers.  A dimer $[i,i+1]$, occupying spaces $i,i+1$, can be in
 $\cF$, if and only if $i$ and $i+1$ are in $F$, where $m$ and
 and $m+1\equiv 1$ are adjacent.  To each tiling $\cF$ corresponds a
 monomial $x^{2l}$, where $l$ is the number of dimers in the tiling
 $\cF$ of $F$.  Then $f(x,S,T)$ is the sum of all monomials
 corresponding to all tilings of $F$.  Note that if $S\cap
 T=\emptyset$ and $S\cup T=\an{m}$ then $f(x,S,T)=1$.
 Furthermore $f(x,S,T)=f(x,T,S)$.  Hence the transfer matrix
 $B(m,\v)$ is a nonnegative symmetric matrix.  The quantity
 $\bar P_1(m,\v)$, defined by (\ref{hbardef1}), is given as the
 logarithm of the spectral radius of $B(m_1,\v)$.
 In numerical computations, we view $\frac{\bar P_1(m,\v)}{m}$
 as an approximation to the pressure $P_2(\v)$.  More
 precisely, one has the upper and lower bounds on the pressure
 which are given by (\ref{imlb}).

 As in \cite{FP}, the matrix $B(m,\v)$ has an automorphism
 group of order $2m$, obtained by rotating the discrete torus
 $T(m)$ and reflecting it.  Thus, to compute the spectral
 radius of $B(m,\v)$, it is enough to compute the spectral
 radius of the nonnegative symmetric matrix $\tilde B(m,\v)$
 whose order is slightly higher than $\frac{2^{m-1}}{m}$.
 See for details \cite[Section 7]{FP}.  \cite[Table 1, page
 517]{FP} gives the dimensions of $\tilde B(m,\v)$ for
 $m=4,\ldots,17$.  We were able to carry out some computations
 on a desk top computer up to $m=17$.

 We first apply our techniques to examine the Baxter
 computations in \cite{Bax1}.  Baxter computes essentially the
 values of the pressure $\pres_2(v):=P_2(v,v)$ and the
 corresponding density of the dimers $p(v):=\frac{\d}{\d v}\pres_2(v)$.
 Recall that the corresponding density entropy $h_2(p(v))$ is
 given by $\pres_2(v)-vp(v)$ (\ref{stadforhdp}).
 Note the following correspondence between the variables in
 \cite{Bax1} and our variables given in
 Section~\ref{sec:weightedmonomerdimertiligs}:
 $$s=e^{v}, \quad \frac{\kappa}{s}=e^{-v+\pres_2(v)},\quad \rho=\frac{p}{2}.$$
 The case $s=v=\infty$ corresponds to the dimer tilings of
 $\Z^2$.  In this case $p=p(\infty)=1$ and $h_2(1)$ has a known
 closed formula due to Fisher \cite{Fis} and Kasteleyn
 \cite{Kas}
 $$h_2(1)=\frac{1}{\pi}\sum_{r=0}^{\infty}\frac{(-1)^r}{(2r+1)^2}=0.29156090\ldots
 .$$
 As in \cite{Bax1} we consider the following 18 values of $s$
 \begin{eqnarray*}
 s^{-1}=0.02,0.05,0.10,0.20,0.30,0.40,0.50,0.60,0.80,\\
 1.00,1.50,2.00,2.50,3.00,3.50,4.00, 4.50,5.00.
 \end{eqnarray*}

 We computed the upper and the lower bounds for $\pres_2(\log s)$ for the above values of $s$,
 using inequalities (\ref{imlb}) for $\v=(\log s,\log s)$ and $m=2,\ldots,17$.
 In these computations we observed that the sequence $\frac{\bar P_1(2r,(\log s,\log s))}{2r}$ is decreasing
 for $r=1,\ldots,8$.  So our upper bound was given by  $\frac{\bar P_1(16,(\log s,\log s))}{16}$ for all 18 values of $s$.
 The lower bound was given by $\frac{\bar P_1(16,(\log s,\log s))-\bar P_1(14,(\log s,\log s))}{2}$ for $s^{-1}=0.02,\ldots,0.3$
 and by  $\bar P_1(17,(\log s,\log s))-\bar P_1(16,(\log s,\log s))$ for other values of $s$.

 The values of Baxter for the pressure were all but two values between the upper and the lower bounds.  In the two
 exceptional values $s^{-1}=1.5,2.0$ Baxter's result were off by $1$ in the last 10th digit.
 As in Baxter computations, the difference between the upper and lower bounds grows bigger as the value of $s$ increases.
 That is, it is harder to compute the precise value of the pressure and its derivative in configurations where the density
 of dimers is high.  This points to the phase transition in the case where $\Z^2$ is tiled by dimers
 only \cite[p'133]{AuP}.  The pressure value for $s^{-1}=0.02$ computed by Baxter has $8$ values.  Our upper and lower bounds give 4 digits of
 precision of the pressure.  For the value $s=1.0$ our computations confirm the first 9 digits of 10 digit Baxter computation.
 (This value of the pressure is equal to the monomer-dimer entropy $h_2$ discussed in \cite{FP}.)
 For the values $s^{-1}=2.0,\ldots,5.0$ our computations gives at least 12 digits of the pressure.

 We also computed the approximate value of the dimer density $p(\log s)=\pres_2'(\log s)$ using
 the following two methods.  The first approximation was obtained by computing the exact derivative of $\frac{\bar P_1(m,(\log s,\log s))}{m}$
 for $m=2,\ldots,14$.  The second approximation was obtained by computing the ratio
 $\frac{\bar P_1(m,(\log (s+t),\log (s+t)))-\bar P_1(m,(\log s,\log s))}{m}$ for $t=10^{-5}$ and $m=2,\ldots,14$.
 It turned out that the values of the numerical
 derivatives for $m=14$ agrees with most values of Baxter computations up to 5 digits, while the values of the exact derivatives agrees only
 up 2 digits with Baxter computations.  Note that to compute the value of $h_2(p(\log s))$ we need the values of $\pres_2(\log s)$
 and $p(\log s)$ (\ref{stadforhdp}).

 We next computed the approximate values of the pressure $\pres_2((v_1,v_2))$ and its partial numerical derivatives
 for $18^2=324$ values.
 The $18$ values of $v_1$ and $v_2$ were chosen in the interval $(-1.61,4.)$.  (These values correspond to the $18$ values of $\log s$
 considered by Baxter.)
 For the lower bound and upper bounds we chose the values of
 \begin{equation}\label{lowupbdpres}
 \frac{\bar P_1(14,(v_1,v_2))-
 \bar P_1(12,(v_1,v_2))}{2} \quad  \textrm{and }\frac{\bar P_1(14,(v_1,v_2))}{14} \textrm{ respectively}.
 \end{equation}
 Follows below the graph of $\frac{\bar P_1(14,(v_1,v_2))}{14}$ and the approximate values of $\bar h_2((p_1,p_2))$, where
 $p_1,p_2$ are the densities of the dimers in the direction $x_1,x_2$ respectively.  The approximate values of $\bar h_2$
 obtained by using the formula
 \begin{eqnarray}\label{H2approx}
 \bar h_2((p_1,p_2))\approx\frac{\bar P_1(14,(v_1,v_2))}{14} -p_1v_1  - p_2v_2 ,\\
 p_1=  \frac{\bar P_1(14,(v_1+t,v_2))-\bar P_1(14,(v_1,v_2))}{14 t}, \\
 p_2= \frac{\bar P_1(14,(v_1,v_2+t))-\bar P_1(14,(v_1,v_2))}{14 t},
 \nonumber
 \end{eqnarray}
 In our computation $t=10^{-4}$.  For more detailed graphs with $42^2=1764$ points see

 \noindent
 http://www2.math.uic.edu/$\sim$friedlan/Pressure17Jun09.pdf

  \begin{figure}[h]
    \includegraphics*[width=8cm,height=8cm]{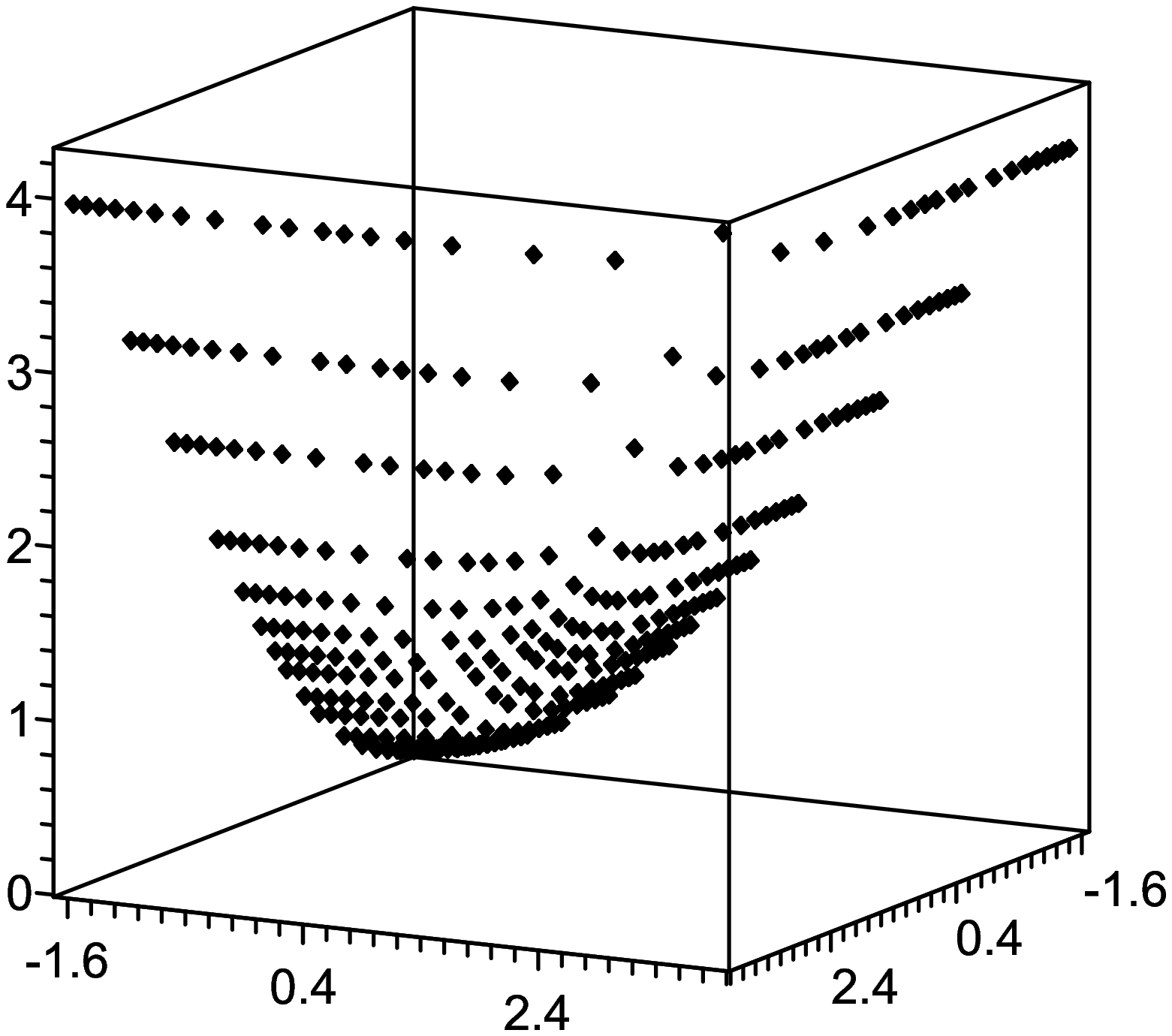}
    \includegraphics*[width=8cm,height=8cm]{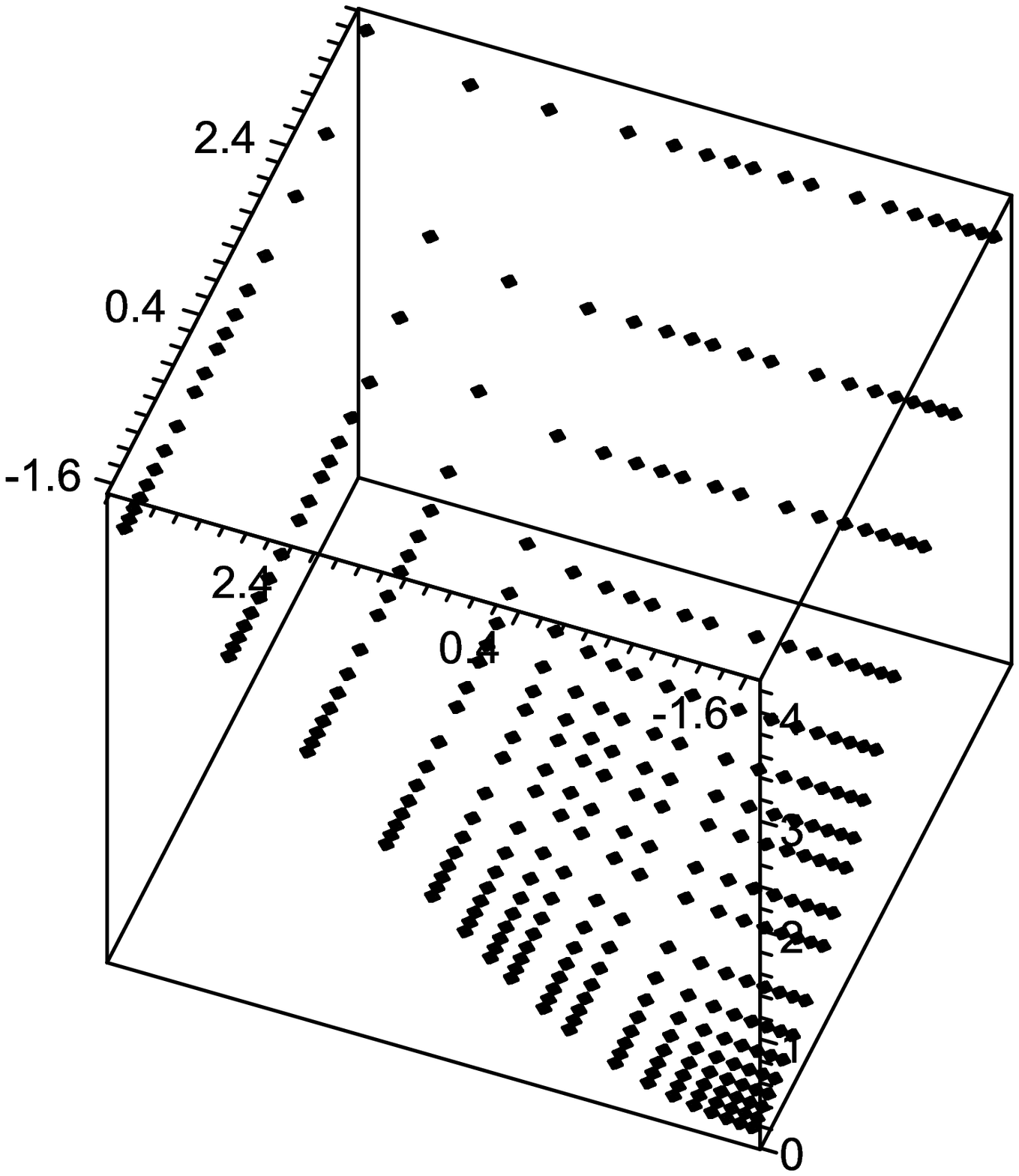}
    \caption{The graph of $\frac{\bar P_1(14,(v_1,v_2))}{14}$ for angles $\theta=28^{o},\varphi=78^o$ and
    $\theta=-159^o,\varphi=42^0$}
    \label{fig:Lpresu24Jun09}
     \end{figure}

  \begin{figure}[h]
    \includegraphics*[width=8cm,height=8cm]{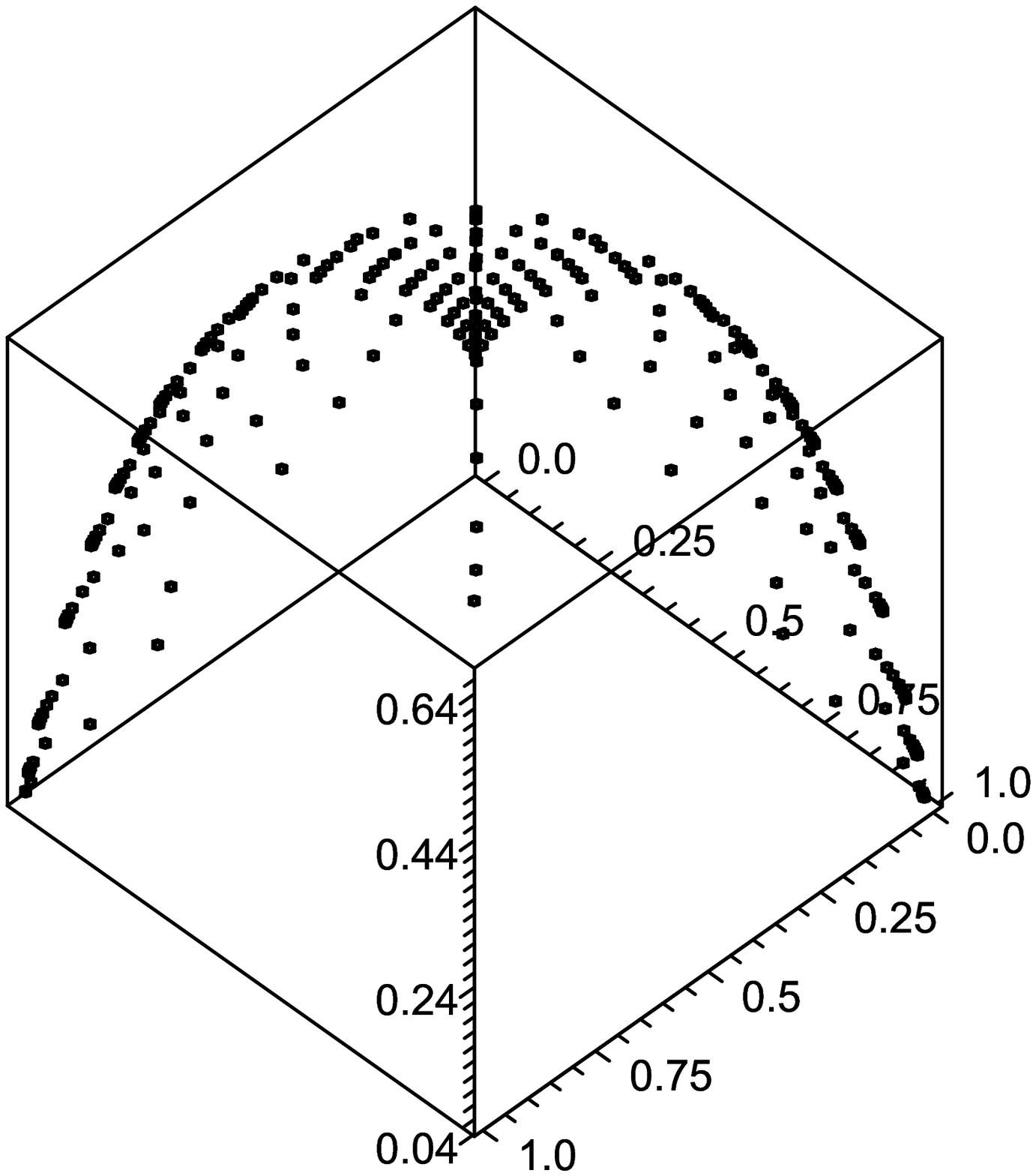}
    \includegraphics*[width=8cm,height=8cm]{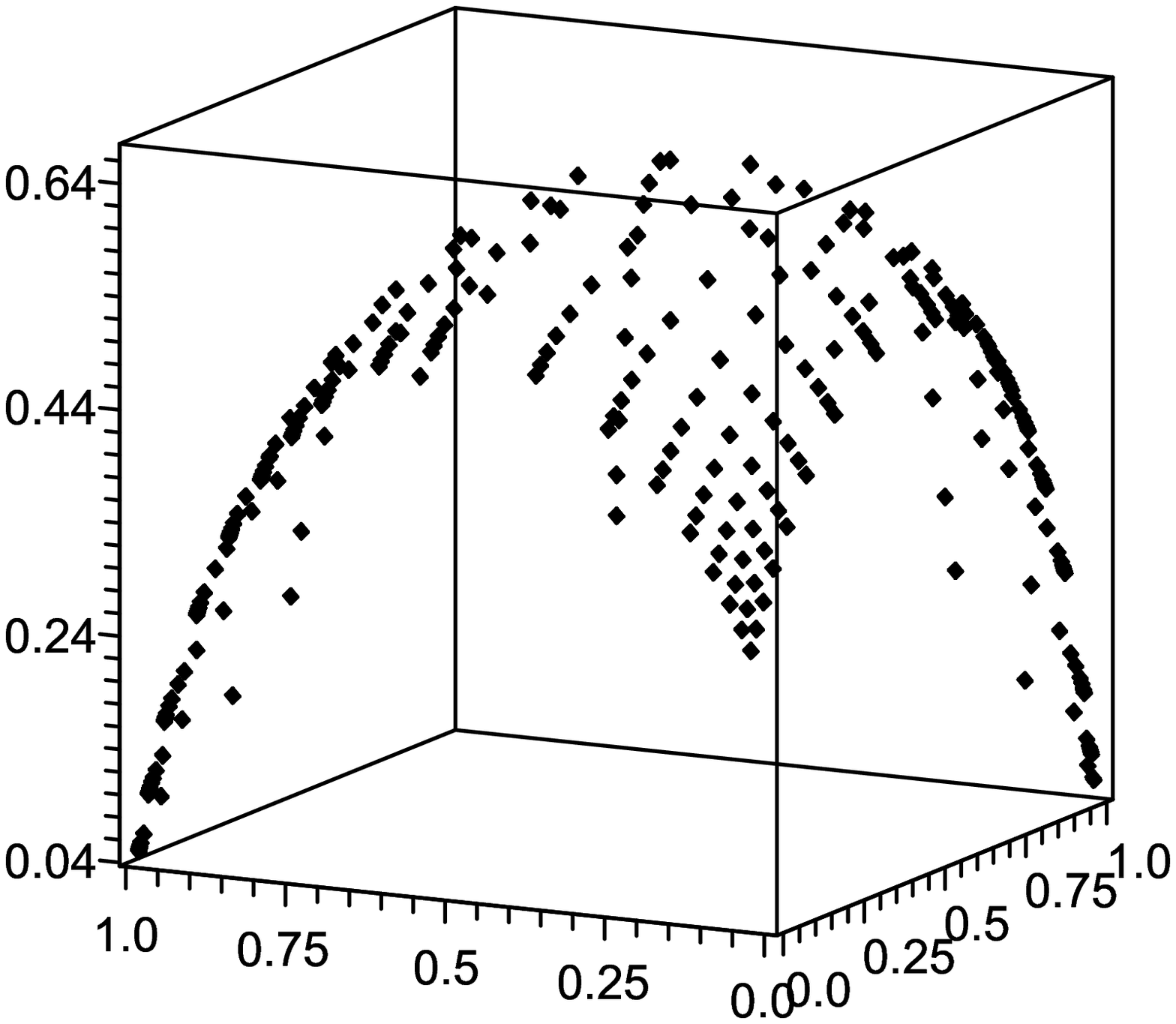}
    \caption{The graph of an approximation of $\bar h_2((p_1,p_2))$ for angles $\theta=45^{o},\varphi=45^o$ and $\theta=-153^o,\varphi=78^o$}
    \label{fig:Llambu24Jun09}
     \end{figure}
 The graph of the pressure $P_2((x_1,x_2))$  is convex and the graph of the density entropy $\bar h_2(x_1,x_2)$ is concave.
 Both graphs look  is symmetric with respect to the line $x_1=x_2$.  In reality this is not the case, since $\bar P_1(m,v_1,v_2)$
 is the pressure of an infinite torus with a basis $m$.  So in direction $x_1$ we have at most $\lfloor \frac{m}{2}\rfloor$ dimers,
 while in the direction $x_2$ we can have an infinite number of dimers.  For $m\ge 10$ the difference
 $\frac{\bar P_1(m,v_1,v_2)-\bar P_1(m,v_2,v_1)}{m}$ is less than $10^{-3}$, which explains the symmetry of our graphs.
 Note that in Figure 2 the densities $p_1,p_2$ satisfy the condition $p_1,p_2\in [0,1], p_1+p_2\in [0,1]$.
 The entropy $\bar h_2$ is in the interval $[0, 0.67]$.

 We also got similar graphs for the lower bound given in (\ref{lowupbdpres})
 and the corresponding analog of the approximation of  $\bar h_2((p_1,p_2))$ given by (\ref{H2approx}).
 These graphs were very similar to the graphs
 of  $\frac{\bar P_1(14,(v_1,v_2))}{14}$ and the approximation of  $\bar h_2((p_1,p_2))$ given by (\ref{H2approx}).

 \bibliographystyle{plain}
 
 \end{document}